\documentclass[english,a4paper,times,round]{article}
\usepackage[T1]{fontenc}
\usepackage[latin9]{inputenc}
\usepackage{xcolor}
\usepackage{pdfcolmk}
\usepackage{verbatim}
\usepackage{textcomp}
\usepackage{url}
\usepackage{amsthm}
\usepackage{amsmath}
\usepackage{amssymb}
\usepackage{graphicx}
\usepackage{esint}
\usepackage[authoryear]{natbib}
\PassOptionsToPackage{normalem}{ulem}
\usepackage{ulem}

\makeatletter

\providecommand{\tabularnewline}{\\}
\providecolor{lyxadded}{rgb}{0,0,1}
\providecolor{lyxdeleted}{rgb}{1,0,0}

\theoremstyle{plain}
\newtheorem{thm}{\protect\theoremname}
  \theoremstyle{plain}
  \newtheorem{prop}[thm]{\protect\propositionname}

\usepackage{fullpage}
\usepackage{clrscode3e}
\usepackage{textcomp}

\newtheorem{algorithm}{Algorithm}

\renewcommand{\For}{\kw{for}\,\Indentmore}

\graphicspath{{figs/}}

\makeatother

\usepackage{babel}
  \providecommand{\propositionname}{Proposition}
\providecommand{\theoremname}{Theorem}

\begin{document}

\title{Sequential Monte Carlo with Highly Informative Observations}

\author{Pierre Del Moral\thanks{P. Del Moral is with the University of New South Wales.}~~and~Lawrence
M. Murray\thanks{L. M. Murray is with the University of Oxford.}}
\maketitle
\begin{abstract}
We propose sequential Monte Carlo (SMC) methods for sampling the posterior
distribution of state-space models under highly informative observation
regimes, a situation in which standard SMC methods can perform poorly.
A special case is simulating bridges between given initial and final
values. The basic idea is to introduce a schedule of intermediate
weighting and resampling times between observation times, which guide
particles towards the final state. This can always be done for continuous-time
models, and may be done for discrete-time models under sparse observation
regimes; our main focus is on continuous-time diffusion processes.
The methods are broadly applicable in that they support multivariate
models with partial observation, do not require simulation of the
backward transition (which is often unavailable), and, where possible,
avoid pointwise evaluation of the forward transition. When simulating
bridges, the last cannot be avoided entirely without concessions,
and we suggest an $\epsilon$-ball approach (reminiscent of Approximate
Bayesian Computation) as a workaround. Compared to the bootstrap particle
filter, the new methods deliver substantially reduced mean squared
error in normalising constant estimates, even after accounting for
execution time. The methods are demonstrated for state estimation
with two toy examples, and for parameter estimation (within a particle
marginal Metropolis--Hastings sampler) with three applied examples
in econometrics, epidemiology and marine biogeochemistry.
\end{abstract}

\section{Introduction}

Consider the multivariate and continuous-time Markov process $X(t)\in\mathbb{R}^{d}$
and parameters $\Theta$. For a sequence of times $t_{0},\ldots,t_{n}$
we write $X_{0:n}\equiv\{X_{0},\ldots,X_{n}\}\equiv\{X(t_{0}),\ldots,X(t_{n})\}$,
and adopt the convention that uppercase symbols denote random variables,
with matching lowercase symbols realisations of them. We have $X_{k}\sim p(dx_{k}\,|\,x_{k-1},\theta)$.
The process may be observed indirectly via some $Y_{n}\sim p(dy_{n}\,|\,x_{n},\theta)$,
or directly as some given initial value $x_{0}$ and final value $x_{n}$.
This setup admits continuous-time models, where $n$ can be made arbitrarily
large and so $t_{1},\ldots,t_{n-1}$ arbitrarily dense, and discrete-time
models with sparse observation, where $n$ cannot be made arbitrarily
large, but where times $t_{1},\ldots,t_{n-1}$ are unobserved.

In the case of indirect observation, the problem of interest is to
simulate $X_{0:n}\sim p(dx_{0:n}\,|\,y_{n},\theta)$ and perhaps estimate
the normalising constant (marginal likelihood) $p(y_{n}\,|\,\theta)$.
We are particularly interested in the case that the observation is
in some sense highly informative on $X_{0:n}$. This might arise when
highly accurate measurements are taken in controlled settings, or
where observations are relatively sparse in time. It will be adequate
in this work for the meaning of ``highly informative'' to remain
qualitative---vague even---but if it were to be quantified it might
be defined as some large divergence (e.g. Kullback--Leibler) of $p(dx_{n}\,|\,\theta)$
from $p(dx_{n}\,|\,y_{n},\theta)$.

In the case of direct observation, the problem of interest is to simulate\emph{
bridges} $X_{1:n-1}\sim p(dx_{1:n-1}\,|\,x_{0},x_{n},\theta)$ and
perhaps estimate the normalising constant (transition density) $p(x_{n}\,|\,x_{0},\theta)$.
This might be seen as the special case of indirect observation with
$p(y_{n}\,|\,x_{n})=\delta(y_{n}-x_{n})$, where $\delta$ is the
Dirac $\delta$ function. We motivate the approach with this special
case, and return to the more general case later.

The simulation of bridges is generally regarded as a difficult problem,
and there is a wealth of literature in the area. Most approaches,
while not necessarily limited to such, begin with a diffusion process
satisfying the Itô stochastic differential equation (SDE)
\[
dX(t)=a(X(t),t,\theta)\,dt+B(X(t),t,\theta)\,dW(t),
\]
where $a(X(t),t,\theta)$ is the drift vector, $B(X(t),t,\theta)$
the diffusion matrix and $W(t)$ a vector of standard Wiener processes.
One group of methods proceeds with a schedule of times $t_{0},\ldots,t_{n}$,
equispaced at a sufficiently small step size that a locally linear--Gaussian
Euler--Maruyama~\citep[\S9.1]{Kloeden1992} discretisation of the
original dynamics makes a credible approximation:
\[
p(x_{k}\,|\,x_{k-1},\theta):=\mathcal{\phi}\left(x_{k-1}+a(x_{k-1},t_{k-1},\theta)\Delta t_{k},\,B(x_{k-1},t_{k-1},\theta)B(x_{k-1},t_{k-1},\theta)^{\top}\Delta t_{k}\right).
\]
Here, $k=1,\ldots,n$, $\Delta t_{k}=t_{k}-t_{k-1}$ and $\phi(\mu,\Sigma)$
is the probability density function of the normal distribution with
mean vector $\mu$ and covariance matrix $\Sigma$. Because of the
convenience of this closed form for $p(x_{k}\,|\,x_{k-1},\theta)$,
the discretised version may be adopted in place of the continuous
version (as an \textsl{a priori} approximation) before proceeding
with inference. This facilitates various importance sampling~\citep[e.g.][]{Pedersen1995,Durham2002},
sequential Monte Carlo (SMC,~e.g. \citealp{Lin2010}) and Markov
chain Monte Carlo (MCMC,~e.g. \citealp{Roberts2001,Elerian2001,Eraker2001,Golightly2006})
methods for simulating bridges.

A caution is warranted around the use of Euler--Maruyama, however.
While the discretisation provides a convenient approximation for pointwise
evaluation of $p(x_{k}\,|\,x_{k-1},\theta)$, it can be unstable for
simulation unless the step size is very small---too small for efficient
computation. In such cases higher-order schemes, such as the Milstein
and semi-implicit schemes~\citep{Kloeden1992}, or those of the Runge--Kutta
family, are to be preferred, and indeed may be essential. For any
given step size $\Delta t$, these higher-order schemes have stability
regions at least as large as that of Euler--Maruyama. For particularly
stiff problems, implicit schemes may also need to be considered.

Rather than discretising to yield a closed-form transition density,
a closed-form Radon--Nikodym derivative can be derived in certain
conditions. This leads to another family of methods for simulating
bridges~\citep[e.g.][]{Clark1990,Delyon2006,Bayer2013,Schauer2013}.
The exact algorithm~\citep{Beskos2006} is yet another alternative,
and has the advantage of not introducing discretisation error, but
its requirements restrict the set of models to which it can be applied,
especially the set of multivariate models.

The estimation of the normalising constant $p(x_{n}\,|\,x_{0},\theta)$
usually falls out as a straightforward expectation in importance sampling
methods, including SMC. It is more difficult with MCMC. A number of
works have focused more closely on this, usually in the context of
obtaining normalising constant estimates for parameter estimation~\citep[e.g.][]{Fearnhead2008,Fearnhead2010a,Sun2013}.

The methods proposed in this work are most similar to the SMC methods
of \citet{DelMoral2005}, and \citet{Lin2010}. The former considers
the simulation of rare events, not bridges, albeit with similar mechanisms.
The latter does consider the simulation of bridges, but the implementation
has some limitations, which this work seeks to ameliorate. In particular,
\citet{Lin2010} requires that pilot samples can be initialised at
time $t_{n}$ and simulated backwards in time to guide those samples
being simulated forwards in time. This is problematic in cases where
the backwards transition cannot be discretised in a numerically stable
way. In contrast, the proposed methods do not require simulation of
the backwards transition. In addition, \citet{Lin2010} does not support
indirect observation, and in the case of multivariate models and direct
observation, does not support partial observation of the state vector
$x_{n}$. The proposed methods accommodate both of these cases. In
the context of indirect observation, the methods are similar to the
auxiliary particle filter~\citep{Pitt1999}, but perform lookahead
steps at multiple intermediate times between observations, rather
than at observation times only.

The remainder of this work describes the proposed methods for sampling
from state-space models with highly informative observations, including
the special case of simulating bridges. The methods also permit estimation
of normalising constants. The basic idea is to simulate particles
forward in time using only the prior $p(dx_{k}\,|\,x_{k-1},\theta)$,
discretised with a higher-order scheme, but to introduce additional
weighting and resampling steps at intermediate times to guide particles
towards the final state. The methods are simple to apply, with the
following properties that make them useful for a broad range of problems:
\begin{enumerate}
\item They work in a multivariate setting.
\item They support partial observation of the state vector $x_{n}$.
\item They do not require that the backward transition can be simulated.
\item They support higher-order discretisations of the forward state process
than that of Euler--Maruyama.
\item They require only that the forward transition can be simulated and
not that its probability density function $p(x_{k}\,|\,x_{k-1},\theta)$
can be evaluated pointwise. This comes with the caveat that in the
special case of simulating bridges, a workaround is needed to approximately
evaluate, or avoid the evaluation of, $p(x_{n}\,|\,x_{n-1},\theta)$.
In one example we use an Euler--Maruyama discretisation to approximately
evaluate, but not simulate, $p(x_{n}\,|\,x_{n-1},\theta)$. In another
we concede an $\epsilon$-ball around $x_{n}$, where $\epsilon$
is commensurate with the discretisation error already surrendered
by the numerical integrator. The latter strategy resembles a simple
Approximate Bayesian Computation (ABC) algorithm~\citep{Beaumont2002}.
\end{enumerate}
Three methods are introduced in this work, all similar but tailored
for slightly different circumstances. §\ref{sec:methods} provides
the formal background to establish the methods. §\ref{sec:implementation}
provides pseudocode for their implementation and discusses design
of the required weight functions. §\ref{sec:results} provides empirical
results for the methods on two toy examples and three applications
in econometrics, epidemiology and marine biogeochemistry. §\ref{sec:discussion}
draws these results together and reports on experiences in tuning
the methods. §\ref{sec:conclusion} concludes.

\section{Methods\label{sec:methods}}

We derive three methods in this section, all quite related, but for
different circumstances of the model and data. References to the parameters,
$\Theta$, are omitted throughout for brevity. Likewise, observation
of the final value $x_{n}$ may be full or partial; the methods support
both cases, but for simplicity no notational distinction is made.

Firstly note:
\begin{eqnarray*}
p(x_{1:n-1}\,|\,x_{0},x_{n}) & = & \frac{p(x_{n}\,|\,x_{n-1})p(x_{1:n-1}\,|\,x_{0})}{p(x_{n}\,|\,x_{0})}\\
 & \propto & p(x_{n}\,|\,x_{n-1})p(x_{1:n-1}\,|\,x_{0}).
\end{eqnarray*}
This forms the basis of the importance sampling method of \citet{Pedersen1995},
proposing from $p(dx_{1:n-1}\,|\,x_{0})$ and weighting with $p(x_{n}\,|\,x_{n-1})$.
The use of the prior as the proposal in this way is myopic of $x_{n}$,
and while often workable, the approach can lead to excessive variance
in importance weights and subsequent computational expense. In \citet{Durham2002},
an alternative proposal is suggested, adjusting the drift term of
the SDE with a linear component to draw the process towards $x_{n}$.
The transition densities $p(x_{k}\,|\,x_{k-1})$ must then be estimated
in order to compute weights, and the low-order Euler--Maruyama discretisation
is employed to achieve this. One, perhaps overlooked, advantage of
\citet{Pedersen1995} is that $p(x_{k}\,|\,x_{k-1})$ does not need
to be evaluated pointwise except for the weight at $k=n$. This means
that, for the purposes of simulation, higher-order discretisation
schemes can be admitted. The methods of this work enjoy the same property.

Further observe:
\begin{equation}
p(x_{1:n-1}\,|\,x_{0})=\prod_{k=1}^{n-1}p(x_{k}\,|\,x_{k-1})\label{eq:markov}
\end{equation}
and
\begin{equation}
\frac{p(x_{n}\,|\,x_{n-1})}{p(x_{n}\,|\,x_{0})}=\prod_{k=1}^{n-1}\frac{p(x_{n}\,|\,x_{k})}{p(x_{n}\,|\,x_{k-1})},\label{eq:weights}
\end{equation}
and define
\[
G_{k}(x_{k-1:k}):=\frac{p(x_{n}\,|\,x_{k})}{p(x_{n}\,|\,x_{k-1})}.
\]
One can then incrementally simulate $p(dx_{k}\,|\,x_{k-1})$ and weight
with $G_{k}(x_{k-1:k})$, for $k=1,\ldots,n-1$. The basis of an SMC
method is to do precisely this, maintaining a population of samples
(particles) and introducing a selection mechanism to resample from
amongst them, according their weights, at each increment. \citet{Lin2010}
does this; the development below follows a similar path, but we will
ultimately suggest a different implementation, and extend the idea
from the special case of sampling bridges to the more general case
of sampling under highly informative observations.
\begin{prop}
\label{prop:exact}For any bounded function $f$ on $(\mathbb{R}^{d})^{n-1}$,
we have
\begin{eqnarray*}
\mathbb{E}\left[f(X_{1:n-1})\,|\,x_{0},x_{n}\right] & = & \mathbb{E}\left[f(X_{1:n-1})\prod_{k=1}^{n-1}G_{k}(X_{k-1:k})\,\middle|\,x_{0}\right]
\end{eqnarray*}
with the weight functions 
\[
G_{k}(x_{k-1:k}):=\frac{p(x_{n}\,|\,x_{k})}{p(x_{n}\,|\,x_{k-1})}.
\]
\end{prop}
\begin{proof}
The conditional density of the random path $X_{1:n-1}$ given initial
state $X_{0}=x_{0}$ and final value $X_{n}=x_{n}$ is
\begin{eqnarray*}
p(x_{1:n-1}\,|\,x_{0},x_{n}) & = & \frac{{\displaystyle p(x_{n}\,|\,x_{n-1})p(x_{1:n-1}\,|\,x_{0})}}{p(x_{n}\,|\,x_{0})}\\
 & = & {\displaystyle \left\{ \prod_{k=1}^{n-1}\frac{p(x_{n}\,|\,x_{k})}{p(x_{n}\,|\,x_{k-1})}\right\} p(x_{1:n-1}\,|\,x_{0})}\\
 & = & \left\{ \prod_{k=1}^{n-1}G_{k}(x_{k-1:k})\right\} p(x_{1:n-1}\,|\,x_{0}).
\end{eqnarray*}
Thus for any bounded function $f$ on $(\mathbb{R}^{d})^{n-1}$ we
have
\begin{eqnarray*}
\mathbb{E}\left[f(X_{1:n-1})\,|\,x_{0},x_{n}\right] & = & \int_{\mathbb{R}^{d(n-1)}}f(x_{1:n-1})p(x_{1:n-1}\,|\,x_{0},x_{n})\,dx_{1:n-1}\\
 & = & \int_{\mathbb{R}^{d(n-1)}}f(x_{1:n-1})\left\{ \prod_{k=1}^{n-1}G_{k}(x_{k-1:k})\right\} p(x_{1:n-1}\,|\,x_{0})\,dx_{1:n-1}\\
 & = & \mathbb{E}\left[f(X_{1:n-1})\prod_{k=1}^{n-1}G_{k}(X_{k-1:k})\,\middle|\,x_{0}\right].
\end{eqnarray*}

\end{proof}
Of course, rarely will the weights $G_{k}(x_{k-1:k})$ be computable
in practice; Proposition \ref{prop:exact} is conceptually appealing,
however, and we can try to imitate it in other circumstances. To do
this, we introduce arbitrary weighting functions that facilitate an
SMC algorithm.
\begin{prop}
\label{prop:importance}For any bounded function $f$ on $(\mathbb{R}^{d})^{n-1}$,
we have 
\[
\begin{array}{l}
\mathbb{E}\left[f(X_{1:n-1})\,|\,x_{0},x_{n}\right]={\displaystyle \frac{\mathbb{E}\left[f(X_{1:n-1})p(x_{n}\,|\,X_{n-1})\frac{q(x_{n}\,|\,x_{0})}{q(x_{n}\,|\,X_{n-1})}\prod_{k=1}^{n-1}H_{k}(X_{k-1:k})\,\middle|\,x_{0}\right]}{\mathbb{E}\left[p(x_{n}\,|\,X_{n-1})\frac{q(x_{n}\,|\,x_{0})}{q(x_{n}\,|\,X_{n-1})}\prod_{k=1}^{n-1}H_{k}(X_{k-1:k})\,\middle|\,x_{0}\right]}}\end{array}.
\]
with the weight functions
\[
H_{k}(x_{k-1:k}):=\frac{q(x_{n}\,|\,x_{k})}{q(x_{n}\,|\,x_{k-1})}
\]
and chosen positive functions $q(x_{n}\,|\,x_{k})$ for $k=0,\ldots,N-1$.\end{prop}
\begin{proof}
Observe that 
\[
p(x_{n}\,|\,x_{0})=\mathbb{E}\left[p(x_{n}\,|\,X_{n-1})\,\middle|\,x_{0}\right]
\]
so that
\begin{eqnarray*}
\mathbb{E}\left[f(X_{1:n-1})\,|\,x_{0},x_{n}\right] & = & \mathbb{E}\left[f(X_{1:n-1})\frac{p(x_{n}\,|\,X_{n-1})}{p(x_{n}\,|\,x_{0})}\,\middle|\,x_{0}\right]\\
 & = & \frac{\mathbb{E}\left[f(X_{1:n-1})p(x_{n}\,\middle|\,X_{n-1})\,|\,x_{0}\right]}{p(x_{n}\,|\,x_{0})}\\
 & = & \frac{\mathbb{E}\left[f(X_{1:n-1})p(x_{n}\,\middle|\,X_{n-1})\,|\,x_{0}\right]}{\mathbb{E}\left[p(x_{n}\,|\,X_{n-1})\,\middle|\,x_{0}\right]}.
\end{eqnarray*}
Introduce
\[
\frac{q(x_{n}\,|\,x_{0})}{q(x_{n}\,|\,x_{n-1})}\prod_{k=1}^{n-1}H_{k}(x_{k-1:k})=1.
\]
We then have
\begin{eqnarray*}
\mathbb{E}\left[f(X_{1:n-1})\,|\,x_{0},x_{n}\right] & = & \frac{\mathbb{E}\left[f(X_{1:n-1})p(x_{n}\,|\,X_{n-1})\frac{q(x_{n}\,|\,x_{0})}{q(x_{n}\,|\,X_{n-1})}\prod_{k=1}^{n-1}H_{k}(X_{k-1:k})\,\middle|\,x_{0}\right]}{\mathbb{E}\left[p(x_{n}\,|\,X_{n-1})\frac{q(x_{n}\,|\,x_{0})}{q(x_{n}\,|\,X_{n-1})}\prod_{k=1}^{n-1}H_{k}(X_{k-1:k})\,\middle|\,x_{0}\right]}.
\end{eqnarray*}

\end{proof}
The results are easily adapted to the case where the state is not
observed exactly, but rather with some noise. We introduce the random
variable $Y_{n}\sim p(dy_{n}\,|\,x_{n})$, which will be observed
in place of $X_{n}$, and a prior distribution over the starting state,
$X_{0}\sim p(dx_{0})$.
\begin{prop}
\label{prop:ssm}For any bounded function $f$ on $\mathbb{R}^{d(n+1)}$,
we have 
\[
\begin{array}{l}
\mathbb{E}\left[f(X_{0:n})\,|\,y_{n}\right]={\displaystyle \frac{\mathbb{E}\left[f(X_{0:n})p(X_{0})p(y_{n}\,|\,X_{n})\frac{r(y_{n}\,|\,X_{0})}{r(y_{n}\,|\,X_{n})}\prod_{k=1}^{n}J_{k}(X_{k-1:k})\right]}{\mathbb{E}\left[p(X_{0})p(y_{n}\,|\,X_{n})\frac{r(y_{n}\,|\,X_{0})}{r(y_{n}\,|\,X_{n})}\prod_{k=1}^{n}J_{k}(X_{k-1:k})\right]}}\end{array}.
\]
with the weight functions
\[
J_{k}(x_{k-1:k}):=\frac{r(y_{n}\,|\,x_{k})}{r(y_{n}\,|\,x_{k-1})}
\]
and chosen positive functions $r(y_{n}\,|\,x_{k})$ for $k=0,\ldots,N$.\end{prop}
\begin{proof}
Observe that 
\[
p(y_{n})=\mathbb{E}\left[p(y_{n}\,|\,X_{n})\right]
\]
so that
\begin{eqnarray*}
\mathbb{E}\left[f(X_{0:n})\,|\,y_{n}\right] & = & \mathbb{E}\left[f(X_{0:n})\frac{p(y_{n}\,|\,X_{n})}{p(y_{n})}\right]\\
 & = & \frac{\mathbb{E}\left[f(X_{0:n})p(y_{n}\,|\,X_{n})\right]}{p(y_{n})}\\
 & = & \frac{\mathbb{E}\left[f(X_{0:n})p(y_{n}\,|\,X_{n})\right]}{\mathbb{E}\left[p(y_{n}\,|\,X_{n})\right]}.
\end{eqnarray*}
Introduce
\[
\frac{r(y_{n}\,|\,x_{0})}{r(y_{n}\,|\,x_{n})}\prod_{k=1}^{n}J_{k}(x_{k-1:k})=1,
\]
We then have
\begin{eqnarray*}
\mathbb{E}\left[f(X_{0:n})\,|\,y_{n}\right] & = & {\displaystyle \frac{\mathbb{E}\left[f(X_{0:n})p(X_{0})p(y_{n}\,|\,X_{n})\frac{r(y_{n}\,|\,X_{0})}{r(y_{n}\,|\,X_{n})}\prod_{k=1}^{n}J_{k}(X_{k-1:k})\right]}{\mathbb{E}\left[p(X_{0})p(y_{n}\,|\,X_{n})\frac{r(y_{n}\,|\,X_{0})}{r(y_{n}\,|\,X_{n})}\prod_{k=1}^{n}J_{k}(X_{k-1:k})\right]}}
\end{eqnarray*}

\end{proof}
We expect this modification to lead to an SMC algorithm that is particularly
useful when observations are highly informative.

\section{Implementation\label{sec:implementation}}

The recursive structure of the weight functions in Propositions \ref{prop:exact}--\ref{prop:ssm}
and Markov property of the process $X(t)$ facilitate SMC algorithms
to compute the expectations of interest. Such algorithms propagate,
weight and resample a population of $N$ particles. We present Algorithms
\ref{alg:algorithm1}--\ref{alg:algorithm3} as pseudocode, corresponding
to Propositions \ref{prop:exact}--\ref{prop:ssm}, respectively.
Where a superscript $i$ appears on the left-hand side of an assignment
in these algorithms, the intended interpretation is ``for all $i\in\left\{ 1,\ldots,N\right\} $''.
The left arrow notation ($\leftarrow$) denotes assignment of the
value on the right to the variable on the left, while the tilde notation
($\sim$) denotes assignment of a draw from the distribution on the
right to the variable on the left. We start with Algorithm \ref{alg:algorithm1},
corresponding to Proposition \ref{prop:exact}:

\begin{algorithm}
{\normalfont
\begin{codebox}
\zi
\zi $x^i_0 \leftarrow x_0$ \Comment initialise
\zi $w^i_0 \leftarrow 1/N$
\zi \For $k = 1,\ldots,n-1$
\zi   \If resampling is triggered \Then
\zi     $a^i_k \sim R(da_k\,|\,w^{1:N}_{k-1})$ \Comment select
\zi     $w^i_k \leftarrow 1/N$
\zi   \Else
\zi     $a^i_k \leftarrow i$
\zi     $w^i_k \leftarrow w^i_{k-1}/\sum_j^N w^j_{k-1}$
      \End
\zi   $x^i_k \sim p(dx_k\,|\,x^{a^i_k}_{k-1})$ \Comment propagate
\zi   $w^i_k \leftarrow p(x_n\,|\,x^i_k) \cdot w^i_k/w^{a^i_k}_{k-1}$ \Comment weight
    \End
\end{codebox}
\label{alg:algorithm1}}

\end{algorithm}Algorithm \ref{alg:algorithm2}, corresponding to
Proposition \ref{prop:importance}, is given below. It assumes that
$q(x_{n}\,|\,x_{n-1}):=p(x_{n}\,|\,x_{n-1})$. Note that Algorithm
\ref{alg:algorithm1} is just the special case of Algorithm \ref{alg:algorithm2}
where $q(x_{n}\,|\,x_{k}):=p(x_{n}\,|\,x_{k})$.

\begin{algorithm}
{\normalfont

\begin{codebox}
\zi
\zi $x^i_0 \leftarrow x_0$ \Comment initialise
\zi $w^i_0 \leftarrow 1/N$
\zi \For $k = 1,\ldots,n-1$
\zi   \If resampling is triggered \Then
\zi     $a^i_k \sim R(da_k\,|\,w^{1:N}_{k-1})$ \Comment select
\zi     $w^i_k \leftarrow 1/N$
\zi   \Else
\zi     $a^i_k \leftarrow i$
\zi     $w^i_k \leftarrow w^i_{k-1}/\sum_j^N w^j_{k-1}$
      \End
\zi   $x^i_k \sim p(dx_k\,|\,x^{a^i_k}_{k-1})$ \Comment propagate
\zi   $w^i_k \leftarrow q(x_n\,|\,x^i_k) \cdot w^i_k/w^{a^i_k}_{k-1}$ \Comment weight
    \End
\end{codebox}
\label{alg:algorithm2}}

\end{algorithm}

At the conclusion of Algorithm \ref{alg:algorithm1} or \ref{alg:algorithm2},
let $b_{n-1}^{i}=i$ and, recursively, $b_{k}^{i}=a_{k+1}^{b_{k+1}^{i}}$.
The indices $b_{k}^{i}$ then establish ancestral lines
\[
x_{0}\rightarrow x_{1}^{b_{1}^{i}}\rightarrow x_{2}^{b_{2}^{i}}\rightarrow\cdots\rightarrow x_{n-2}^{b_{n-2}^{i}}\rightarrow x_{n-1}^{b_{n-1}^{i}}\rightarrow x_{n}.
\]
Because SMC methods are a particular case of mean field particle methods~\citep{DelMoral2004},
these may be used to compute expectations of the forms that appear
in Propositions \ref{prop:exact} and \ref{prop:importance}:
\begin{eqnarray*}
\frac{\sum_{i=1}^{N}w_{n-1}^{i}f\left(x_{1}^{b_{1}^{i}},\ldots,x_{n-1}^{b_{n-1}^{i}}\right)}{\sum_{i=1}^{N}w_{n-1}^{i}} & \underset{N\uparrow\infty}{\longrightarrow} & \mathbb{E}\left[f(X_{1:n-1})\,|\,x_{0},x_{n}\right].
\end{eqnarray*}
The denominator on the left is also an estimate of the normalising
constant:
\begin{eqnarray}
\sum_{i=1}^{N}w_{n-1}^{i} & \underset{N\uparrow\infty}{\longrightarrow} & \mathbb{E}\left[p(x_{n}\,|\,X_{n-1})\,\middle|\,x_{0}\right]\label{eq:normalising-constant1}\\
 & = & p(x_{n}\,|\,x_{0}).\nonumber 
\end{eqnarray}
Note that as Algorithms \ref{alg:algorithm1} and \ref{alg:algorithm2}
normalise the weights $w_{n-1}^{1:N}$ after the selection step but
before the weighting step, no factor of $1/N$ appears outside the
summation in (\ref{eq:normalising-constant1}).

It is often the case that the transition density $p(x_{n}\,|\,x_{n-1})$
does not have a convenient closed form for pointwise evaluation, so
that the last line of Algorithm \ref{alg:algorithm2} cannot be evaluated.
In such cases one of two approaches might be considered. In the first
approach, the sequence of times $t_{1},\ldots,t_{n-1}$ might be set
so that the last interval, $\Delta t_{n}=t_{n}-t_{n-1}$, is sufficiently
small for the Euler--Maruyama approximation of $p(x_{n}\,|\,x_{n-1})$
to be credible. Because this last transition is evaluated but not
simulated (much less simulated repeatedly with accumulating error),
the stability issues of the Euler--Maruyama discretisation will not
manifest. In the second approach, an observation model might be constructed
with an $\epsilon$-ball around $x_{n}$, where $\epsilon$ is commensurate
with the discretisation error already inherent in the numerical integrator.
We do this in the SIR example of §\ref{sec:results}.

This second approach yields an algorithm resembling ABC~\citep{Beaumont2002}.
In a simple ABC algorithm, one would simulate a path $x'_{1:n}\sim p(dx_{1:n}\,|\,x_{0})$
and accept it if $\rho(x'_{n},x_{n})\leq\epsilon$ for some distance
function $\rho$ and error threshold $\epsilon$. The SMC component
of the proposed method marginalises over multiple such paths. If we
define the unnormalised density
\[
p(x_{n}\,|\,x_{n-1}):=\begin{cases}
1, & \text{if}\,\rho(x'_{n},x_{n})\leq\epsilon\\
0, & \text{otherwise}
\end{cases},
\]
then the estimate of the normalising constant (\ref{eq:normalising-constant1})
is also an estimate of the acceptance probability of $\theta'$. It
is worth stressing that the SMC component in this case is used in
a very different way to ABC SMC methods in the spirit of e.g. \citet{Sisson2007,Beaumont2009,DelMoral2012a,Peters2012}.
In these works, SMC is used over parameters, here it is used over
the state. It could be coupled with MCMC (as in particle MCMC, \citealt{Andrieu2010})
or another level of SMC (as in SMC$^{2}$, \citealt{Chopin2013})
for parameter estimation, however.

Algorithm \ref{alg:algorithm3}, corresponding to Proposition \ref{prop:ssm},
is a slight variation on Algorithm \ref{alg:algorithm2}, as $x_{0}$
and $x_{n}$ are no longer fixed and $y_{n}$ is introduced. It assumes,
sensibly, that $r(y_{n}\,|\,x_{n}):=p(y_{n}\,|\,x_{n})$.

\begin{algorithm}
{\normalfont

\begin{codebox}
\zi
\zi $x^i_0 \sim p(dx_0)$ \Comment initialise
\zi $w^i_0 \leftarrow 1/N$
\zi \For $k = 1,\ldots,n$
\zi   \If resampling is triggered \Then
\zi     $a^i_k \sim R(da_k\,|\,w^{1:N}_{k-1})$ \Comment select
\zi     $w^i_k \leftarrow 1/N$
\zi   \Else
\zi     $a^i_k \leftarrow i$
\zi     $w^i_k \leftarrow w^i_{k-1}/\sum_j^N w^j_{k-1}$
      \End
\zi   $x^i_k \sim p(dx_k\,|\,x^{a^i_k}_{k-1})$ \Comment propagate
\zi   $w^i_k \leftarrow r(y_n\,|\,x^i_k) \cdot w^i_k/w^{a^i_k}_{k-1}$ \Comment weight
    \End
\end{codebox}
\label{alg:algorithm3}}

\end{algorithm}

At the conclusion of Algorithm \ref{alg:algorithm3}, the indices
$b_{k}^{i}$, defined as before, establish ancestral lines
\[
x_{0}^{b_{0}^{i}}\rightarrow x_{1}^{b_{1}^{i}}\rightarrow x_{2}^{b_{2}^{i}}\rightarrow\cdots\rightarrow x_{n-2}^{b_{n-2}^{i}}\rightarrow x_{n-1}^{b_{n-1}^{i}}\rightarrow x_{n}^{b_{n}^{i}}.
\]
These may be used to compute expectations of the form that appears
in Proposition \ref{prop:ssm}:
\begin{eqnarray*}
\frac{\sum_{i=1}^{N}w_{n}^{i}f\left(x_{0}^{b_{0}^{i}},\ldots,x_{n}^{b_{n}^{i}}\right)}{\sum_{i=1}^{N}w_{n}^{i}} & \underset{N\uparrow\infty}{\longrightarrow} & \mathbb{E}\left[f(X_{0:n})\,|\,y_{n}\right].
\end{eqnarray*}
The denominator on the left is also an estimate of the normalising
constant:
\begin{eqnarray}
\sum_{i=1}^{N}w_{n}^{i} & \underset{N\uparrow\infty}{\longrightarrow} & \mathbb{E}\left[p(y_{n}\,|\,X_{n})\right]\label{eq:normalising-constant2}\\
 & = & p(y_{n}).\nonumber 
\end{eqnarray}
Note that as Algorithm \ref{alg:algorithm3} normalises the weights
$w_{n}^{1:N}$ after the selection step but before the weighting step,
no factor of $1/N$ appears outside the summation in (\ref{eq:normalising-constant2}).

Algorithm \ref{alg:algorithm3} treats the case where there is only
a single observation, at time $t_{n}$. This is straightforwardly
extended to a time series of observations, where the algorithm is
repeated, but removing the first two lines from the second and subsequent
iterations; the current particles and their weights are maintained
instead. This is then a particle filter, with the addition of intermediate
times between observations where additional weighting and resampling
is performed to guide particles towards the next state.

\subsection{Intermediate weighting}

What remains is the selection of appropriate functions $q$ and $r$
in Algorithms \ref{alg:algorithm2} and \ref{alg:algorithm3}. For
good performance, we should seek $q(x_{n}\,|\,x_{k})\approx p(x_{n}\,|\,x_{k})$,
so that Algorithm \ref{alg:algorithm2} approximates Algorithm \ref{alg:algorithm1}
as closely as possible, and $r(y_{n}\,|\,x_{k})\approx p(y_{n}\,|\,x_{k})$,
much like the use of lookahead~\citep{Lin2013} strategies for stage
one weights in an auxiliary particle filter~\citep{Pitt1999}. As
$q(x_{n}\,|\,x_{k})$ and $r(y_{n}\,|\,x_{k})$ play a similar role
to the proposal distribution in importance sampling, we should also
prefer that their tails are not too tight with respect to $p(x_{n}\,|\,x_{k})$
or $p(y_{n}\,|\,x_{k})$, respectively. For reasons of computational
expediency, we suppose that the weight functions are to be selected
\textsl{a priori}. Choosing some parametric form, we might choose
either to fit the function to simulations of the prior model, or to
the data set. We utilise both approaches in the examples of §\ref{sec:results}.

The implementation in \citet{Lin2010} uses a kernel density estimate
of $p(x_{k}\,|\,x_{n})\propto p(x_{n}\,|\,x_{k})p(x_{k})$, obtained
by propagating pilot particles backwards from time $t_{n}$ to time
$t_{k}$, initialising each at $x_{n}$. This is problematic for the
constraints we have given ourselves: it requires that $x_{n}$ is
fully observed in order to initialise each particle, it does not support
indirect observation, and we do not wish to assume that the backwards
transition can be simulated in a numerically stable way.

For diffusion processes, discretisations that provide a closed-form
transition density may be useful. As even very approximate weight
functions may have some utility, an Euler--Maruyama discretisation
might prove useful. More sophisticated Gaussian approximations such
as a linear noise approximation~\citep[p. 258]{VanKampen2007} may
also be useful, if more computationally expensive.

We have found that a generally useful approach is to fit a Gaussian
process to each observed time series and then construct weight functions
based on these;
\[
X(t)\sim\mathcal{GP}\left(\mu(t),C(\Delta t)\right),
\]
with mean function $\mu(t)$ and covariance function $C(\Delta t)$.
This affords a great deal of flexibility in the design of weight functions,
accommodating arbitrary mean functions to capture nonlinear drifts,
and a variety of covariance functions to capture local behaviour and
smoothness. In the examples of §\ref{sec:results} it has not been
necessary to be too clever to obtain good results: the mean function
is always $\mu(t)=0$ and the covariance function of a squared exponential
form, parameterised by $\alpha$ and $\beta$;
\[
C(\Delta t)=\alpha\exp\left(-\frac{1}{2\beta}(\Delta t)^{2}\right).
\]
The parameters $\alpha$ and $\beta$ are set to their maximum likelihood
estimates, obtained offline. One can, of course, imagine more sophisticated
mean and covariance functions---Gaussian processes being very flexible
in this regard---but we have found this simple choice adequate for
the examples here. The functions $q(x_{n}\,|\,x_{k})$ are then constructed
by conditioning the Gaussian process on the current state, taking
\[
q(x_{n}\,|\,x_{k}):=\phi(\mu_{k},\sigma_{k}^{2})
\]
with
\begin{eqnarray*}
\mu_{k} & = & \frac{C(t_{n}-t_{k})}{\alpha}x_{k}\\
\sigma_{k}^{2} & = & \alpha-\frac{C(t_{n}-t_{k})^{2}}{\alpha}.
\end{eqnarray*}
Conditioning on the current state only, and not the full state history,
is a computational concession, preserving linear complexity in the
number of particles $N$. We have found this sufficient for the examples
in this work. If necessary, one might consider conditioning on some
fixed number of previous states, preserving the same linear complexity.
In the case of indirect observation, a (possibly approximate) Gaussian
observation model of 
\[
Y(t)\sim\mathcal{N}(X(t),\varsigma^{2}(t)),
\]
for some variance $\varsigma^{2}(t)$, would suggest weight functions
of
\[
r(y_{n}\,|\,x_{k}):=\phi(\mu_{k},\sigma_{k}^{2}+\varsigma_{n}^{2}).
\]

We may be concerned that these tight-tailed Gaussians are too narrow
as weight functions, and may be too aggressive in particle selection.
A simple precaution is to inflate the variance by some constant factor,
and we do this in the examples of §\ref{sec:results}. One might also
consider heavier-tailed functions such as the Student $t$. Another
limitation of the Gaussian process formulation is that it may be inadequate
for capturing multimodal transition densities. This occurs in the
periodic drift example of Section \ref{sec:results}, and we propose
a bespoke function in that case.

\subsection{Intermediate resampling}

The intermediate resampling steps require some consideration. In the
first instance, they introduce additional computational expense. In
the second they may introduce additional variance in normalising constant
estimates~\citep{Pitt2002,Lee2008}. In the case of computational
cost, the numerical integration of diffusion processes (required to
propagate particles forward) will typically be much more expensive
than the steps required to resample. We might assume, then, that the
additional resampling adds little to overall cost. At any rate an
adaptive resampling scheme mitigates both issues. 

A simple adaptive resampling scheme is based on the \emph{effective
sample size} (ESS), which for the weight vector $w_{k}^{1:N}$ at
time $t_{k}$ is~\citep{Liu1995}
\[
\textrm{ESS}(w_{k}^{1:N})=\frac{\left(\sum_{i}^{N}w_{k}^{i}\right)^{2}}{\sum_{i}^{N}(w_{k}^{i})^{2}}.
\]
Resampling is then only triggered if this quantity falls below some
threshold. We do this in the experimental results of §\ref{sec:results},
and find that the net effect of the additional resampling steps is
beneficial.

When an adaptive scheme such as this is used, the increase in variance
should be constant with respect to the number of intermediate times.
Intuitively, this is clear from (\ref{eq:weights}): the accumulated
weight at some time $t_{k}$ is always $p(x_{n}\,|\,x_{k})/p(x_{n}\,|\,x_{0})$,
regardless of the preceding time schedule $t_{1},\ldots,t_{k-1}$.
Rather than determining the accumulated weight, the time schedule
determines the times at which resampling should be considered. Note
that, if resampling is never triggered, all the additional weights
cancel. Algorithm \ref{alg:algorithm2} then becomes the method of
\citet{Pedersen1995}, while Algorithm \ref{alg:algorithm3}, iterated,
becomes the bootstrap particle filter.

\section{Experiments\label{sec:results}}

We use five different examples to demonstrate the methods:
\begin{description}
\item [{OU}] a linear--Gaussian Ornstein--Uhlenbeck process fit to simulated
data, without parameter estimation~\citep[c.f.][]{Sun2013},
\item [{FFR}] a linear--Gaussian Ornstein--Uhlenbeck process fit to Federal
Funds Rate data, with parameter estimation~\citep[c.f.][]{Ait-Sahalia1999},
\item [{PD}] a nonlinear periodic drift process fit to simulated data,
without parameter estimation~\citep[c.f.][]{Beskos2006,Lin2010},
\item [{SIR}] a multivariate and nonlinear susceptible/infected/recovered
compartmental model used in epidemiology, fit to influenza data from
a boarding school~\citep{Anonymous1978}, with parameter estimation~\citep[c.f.][]{Ross2009},
and
\item [{NPZD}] a multivariate and nonlinear nutrient/phytoplankton/zooplankton/detritus
model~\citep{Parslow2013} used in marine biogeochemistry, fit to
Ocean Station P data~\citep{Matear1995}, with parameter estimation~\citep[c.f.][]{Parslow2013,Murray2013a}.
\end{description}
The OU and PD examples are toy studies used to illustrate the methods,
while the FFR, SIR and NPZD examples are applied problems using real
data sets. The SIR example has additional interest for the ABC-like
approach taken.

Experiments are conducted using the LibBi software~\citep{Murray2013b},
in which the methods have been implemented under the name \emph{bridge
particle filter}. We use this name henceforth. Each example is available
as a separate LibBi package, available from \url{www.libbi.org}.
The bootstrap particle filter, as implemented in LibBi, is used for
comparison, noting that with no resampling at intermediate times,
it reduces to the method of \citet{Pedersen1995}.

Configurations for all experiments are summarised in Table \ref{tab:config}
and detailed in the text. Using LibBi, it is straightforward to run
both the bootstrap and bridge particle filters across multiple threads
on a central processing unit (CPU), with or without the use of SSE
vector instructions, or on a graphics processing unit (GPU). Table
\ref{tab:config} also documents the chosen hardware configuration
for each example, chosen for fastest execution time after some pilot
runs.

To compare methods, we use a number of metrics based on normalising
constant estimates, which are further scaled by execution time for
a fair computational comparison. For a set of $Z$ (see Table \ref{tab:config}
for specifics) normalising constant estimates, $z^{1:Z}$, obtained
after corresponding execution times $t^{1:Z}$, we define the metrics:
\begin{align}
\text{MSE}(\log z^{1:Z})^{-1} & \cdot\text{Mean}(t^{1:Z})^{-1}\label{eqn:metric-mse}\\
\text{ESS}(z^{1:Z}) & \cdot\text{Mean}(t^{1:Z})^{-1}\label{eqn:metric-ess}\\
\text{CAR}(z^{1:Z}) & \cdot\text{Mean}(t^{1:Z})^{-1}.\label{eqn:metric-car}
\end{align}
Here, $\text{Mean(}t^{1:Z})$ is the sample mean of $t^{1:Z}$. $\text{MSE}(z^{1:Z})$
is the mean-squared error of $z^{1:Z}$:
\[
\text{MSE}(z^{1:Z})=\frac{1}{Z}\sum_{i=1}^{Z}(z^{i}-z^{*})^{2},
\]
where $z^{*}$ is the true normalising constant. If $z^{*}$ is not
known, the best available estimate is substituted as its ``true''
value. This will be the estimate obtained from a bootstrap particle
filter\footnote{Our results indicate that the bridge particle filter should typically
give a better estimate, but we avoid basing the truth on the method
to be validated.} using a great many particles (see Table \ref{tab:config} for specifics).
$\text{ESS}(z^{1:Z})$ is the effective sample size of $z^{1:Z}$~\citep{Liu1995}:
\[
\text{ESS}(z^{1:Z})=\frac{\left(\sum_{i=1}^{Z}z^{i}\right)^{2}}{\sum_{i=1}^{Z}(z^{i})^{2}}.
\]
$\text{CAR}(z^{1:Z})$ is the \emph{conditional acceptance rate} of
$z^{1:Z}$~\citep{Murray2013a}:
\[
\text{CAR}(z^{1:Z})=\frac{1}{Z}\left(2\sum_{i=1}^{Z}c^{i}-1\right),
\]
where $c^{i}$ is the sum of the $i$th smallest elements of $z^{1:Z}$.
Higher values are favoured for all three of the metrics (\ref{eqn:metric-mse}--\ref{eqn:metric-car}).

The appropriate metric for comparison between methods depends on the
motivation for estimating the normalising constant. If the estimate
itself is of interest, such as to compute evidence for model comparison,
then the MSE-based metric (\ref{eqn:metric-mse}) is appropriate.
If the estimate is to be used as a weight in some importance sampling
scheme, then the ESS-based metric (\ref{eqn:metric-ess}) is most
appropriate, as the ESS approximates the equivalent number of unweighted
samples. If the estimate is instead to be used in some pseudo-marginal
MCMC scheme, such as particle marginal Metropolis--Hastings (PMMH,~\citealt{Andrieu2010}),
then the CAR-based metric (\ref{eqn:metric-car}) is most appropriate.
CAR is an estimate of the long-term acceptance rate of a Metropolis
chain that makes uniform proposals from $Z$ states with posterior
density proportional to the elements of $z^{1:Z}$~\citep{Murray2013a}.
When exact likelihoods are computed, all elements of $z^{1:Z}$ are
the same, and the CAR is one; in all other cases its difference from
one represents the loss of using an estimated likelihood. Both ESS
and CAR are sensitive to the high tail of $z^{1:Z}$, and reduce substantially
in the presence of high outliers. They capture the dramatic loss in
efficiency of importance and MCMC samplers in such circumstances.
This is a particular risk when choosing weight functions for the bridge
particle filter that are too tight. The MSE does not capture the implications
of such outliers.

The OU, FFR, PD and NPZD examples use a standard set of experiments
for comparing the bootstrap and bridge particle filters. The SIR example
does not, as the bootstrap particle filter could not be configured
to work reliably on it for similar tests. For the toy OU and PD examples,
parameters are fixed and we simulate 16 data sets. For the FFR and
NPZD examples, the data is fixed (a real-world data set), and we simulate
16 parameter sets from the prior for testing. We configure both the
bootstrap and bridge particle filters with the number of particles
set to each of $N\in\{2^{5},2^{6},\ldots,2^{10}\}$. Using all combinations
of the 16 data or parameter sets and six different settings for $N$
gives 96 experiments in total. The three metrics are computed for
each experiment, for a total of 288 comparisons on each example. The
configurations for these tests are given in Table \ref{tab:config}.

The FFR, SIR and NPZD examples use real data sets. We perform parameter
estimation in these scenarios using PMMH. We initialise two Markov
chains, one using a bootstrap particle filter, and one using the bridge
particle filter. Both are initialised to the same initial state, which
has been obtained from a pilot run sufficiently long to have converged
to the posterior distribution. Both use the same proposal distribution,
tuned by hand from the same pilot run. The chains are compared using
acceptance rate and effective sample size. The acceptance rate is
a suitable comparison because the proposal distribution is the same
for both chains, and a higher acceptance rate indicates less variability
in normalising constant estimates. The effective sample size used
is that given in \citet[p. 99]{Kass1998}. This is different to that
defined above, as it is intended for assessing the output of MCMC
rather than that of importance sampling. It is defined as: 
\[
\mathrm{ESS_{MCMC}}(\theta^{1:N_{\theta}})=\frac{N_{\theta}}{1+2\sum_{k=1}^{\infty}R_{\theta}(k)},
\]
where $N_{\theta}$ is the length of the chain and $R_{\theta}(k)$
its lag-$k$ autocorrelation. In practice, $R_{\theta}(k)$ must be
estimated from the chain itself, and the infinite sum truncated at
some finite $k$, after which $R_{\theta}(k)$ is assumed to be zero.
When there is more than one parameter, $\mathrm{ESS_{MCMC}}(\theta^{1:N_{\theta}})$
is computed for each separately, and the smallest value reported.

\begin{table}[tp]
{\footnotesize\begin{centering}
\begin{tabular}{lrrrrr}
\textbf{Model} & \textbf{OU} & \textbf{FFR} & \textbf{PD} & \textbf{SIR} & \textbf{NPZD}\tabularnewline
\hline 
Simulation time step & 0.01 & 0.01 & 0.075 & Adaptive & Adaptive\tabularnewline
Bridge time step & 0.1 & 0.1 & 1 & 0.01 & 1\tabularnewline
Bridge type & Exact & Exact & Parametric & $\mathcal{GP}$ & $\mathcal{GP}$\tabularnewline
 &  &  &  &  & \tabularnewline
\textbf{Data set} &  &  &  &  & \tabularnewline
\textnumero~observations & 100 & 300 & 100 & 14 & 227\tabularnewline
Observation time step & 1 & 1 & 30 & 1 & Irregular\tabularnewline
 &  &  &  &  & \tabularnewline
\textbf{Normalising constant experiments} &  &  &  &  & \tabularnewline
\textnumero~data sets & 16 & 1 & 16 & - & 1\tabularnewline
\textnumero~parameter sets & 1 & 16 & 1 & - & 16\tabularnewline
\textnumero~particles ($N$) & $2^{5},\ldots,2^{10}$ & $2^{5},\ldots,2^{10}$ & $2^{5},\ldots,2^{10}$ & - & $2^{5},\ldots,2^{10}$\tabularnewline
Total \textnumero~experiments & 96 & 96 & 96 & - & 96\tabularnewline
\textnumero~repetitions on each experiment ($Z$) & $2^{12}$ & $2^{12}$ & $2^{12}$ & - & $2^{10}$\tabularnewline
\textnumero~particles ($N$) for ``true'' log-likelihood & - & - & $2^{20}$ & - & $2^{20}$\tabularnewline
 &  &  &  &  & \tabularnewline
\textbf{Parameter estimation experiments} &  &  &  &  & \tabularnewline
\textnumero~MCMC steps after burn-in ($N_{\theta}$) & - & $1\times10^{5}$ & - & $5\times10^{4}$ & $1\times10^{5}$\tabularnewline
\textnumero~particles ($N$) & - & $2^{8}$ & - & $2^{14}$ & $2^{6}$\tabularnewline
Maximum lag for $\mathrm{ESS}_{\textrm{MCMC}}$ & - & 250 & - & 2000 & 5000\tabularnewline
Bootstrap $\mathrm{ESS}_{\textrm{MCMC}}$ & - & $^{*}$ & - & 72.7 & 47.3\tabularnewline
Bridge $\mathrm{ESS}_{\textrm{MCMC}}$ & - & 700.0 & - & 304.7 & 57.0\tabularnewline
Bootstrap acceptance rate (\%) & - & $^{*}$ & - & 3.9 & 14.8\tabularnewline
Bridge acceptance rate (\%) & - & 21.4 & - & 15.1 & 15.1\tabularnewline
 &  &  &  &  & \tabularnewline
\textbf{Configuration} &  &  &  &  & \tabularnewline
\textnumero~CPU threads & 1 & 4 & 1 & 2 & 4\tabularnewline
SSE instructions used? & No & No & No & No & Yes\tabularnewline
GPU used? & No & No & No & Yes & No\tabularnewline
Floating point precision & Double & Double & Double & Single & Double\tabularnewline
Relative ESS threshold for resampling & 0.5 & 0.5 & 0.5 & 0.5 & 0.5\tabularnewline
\hline 
\end{tabular}
\par\end{centering}}

\vspace{3mm}

$^{*}$ {\small{}The bootstrap particle filter consistently degenerates
on the FFR example, and no results could be obtained.}\protect\caption{Experimental configurations and some results for all examples.\label{tab:config}}
\end{table}

\subsection{Ornstein--Uhlenbeck (OU) process}

Consider the Ornstein--Uhlenbeck process satisfing the following Itô
SDE:
\begin{equation}
dX(t)=\left(\theta_{1}-\theta_{2}X(t)\right)\,dt+\theta_{3}\,dW(t),\label{eq:ou-process}
\end{equation}
with parameters $\theta_{1}=0.0187$, $\theta_{2}=0.2610$ and $\theta_{3}=0.0224$,
as obtained in \citet{Ait-Sahalia1999} and used in \citet{Sun2013}.
For step size $\Delta t$, the transition density is~\citep{Sun2013}
\[
p\left(x(t+\Delta t)\,|\,x(t)\right)=\phi\left(\mu(\Delta t),\sigma^{2}(\Delta t)\right),
\]
with
\begin{eqnarray*}
\mu(\Delta t) & = & \frac{\theta_{1}}{\theta_{2}}+\left(x(t)-\frac{\theta_{1}}{\theta_{2}}\right)\exp\left(-\theta_{2}\Delta t\right)\\
\sigma^{2}(\Delta t) & = & \frac{\theta_{3}^{2}}{2\theta_{2}}\left(1-\exp\left(-2\theta_{2}\Delta t\right)\right).
\end{eqnarray*}
Because the transition density is known explicitly for all $\Delta t$,
no approximation of it is required, and Algorithm \ref{alg:algorithm1}
can be applied.

We first consider sampling between the initial value $x(0)=0$ and
final value $x(1)=0.15$, applying both the bootstrap and the bridge
particle filters. The filters are configured as in Table \ref{tab:config}.
The results are shown in Figure \ref{fig:LinearDriftBridge-filter}.
Clearly the bridge particle filter produces a more satisfying result,
with the additional weighting and resampling steps guiding particles
towards the final value. Note that the performance of the bootstrap
particle filter can be made arbitrarily poor on this example by reducing
the discretisation time step, pushing the value of $x(1)=0.15$ further
into the tails of the last transition density. This sensitivity of
the bootstrap particle filter to the discretisation time step would
seem undesirable. If the bridge particle filter is allowed to resample
after each step, it is not so sensitive.

Next, we compare the normalising constant estimates of the bootstrap
and bridge particle filters using the three metrics introduced above.
We generate 16 data sets, each constructed by simulating the model
forward from $x(0)=0$ for 100 time units and taking the state at
times $1,2,\ldots,100$. The number of particles is set variously
to $N\in\{2^{5},2^{6},\ldots,2^{10}\}$. Each unique pair of a data
set and an $N$ constitutes an experiment, for 96 experiments in total.
The bootstrap and bridge particle filters are applied to each experiment
4096 times, each time producing an estimate of the normalising constant.
From these estimates, each of the three metrics is computed. For computing
the MSE-based metric, the true normalising constant is used, this
being readily computed as the model is linear and Gaussian. Results
are in Figure \ref{fig:LinearDriftBridge-metrics}. From this we see
that the bridge particle filter outperforms the bootstrap particle
filter in the great majority of comparisons, very often substantially
so.

\begin{figure}[p]
\begin{centering}
\includegraphics[width=1\textwidth]{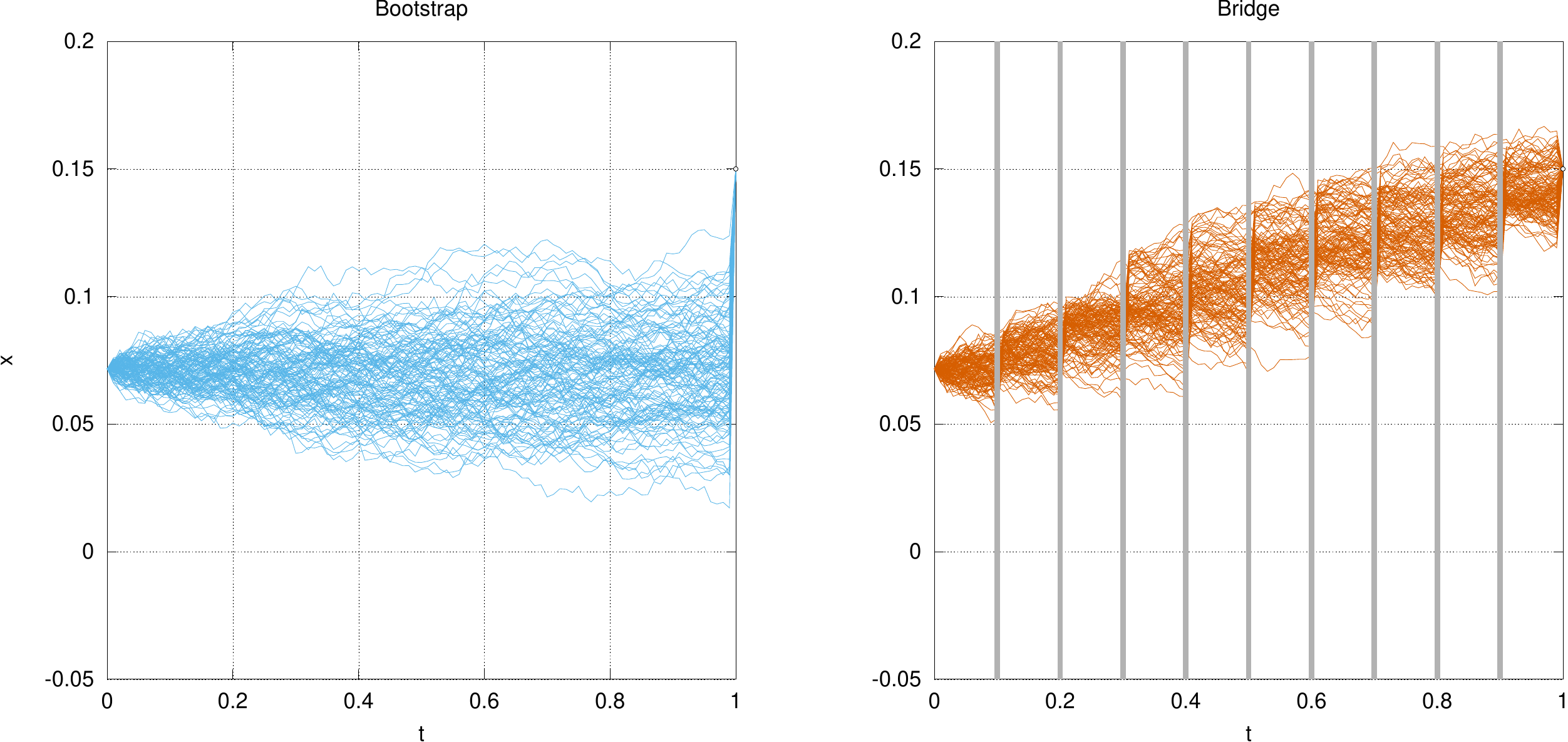}
\par\end{centering}

\protect\caption{Comparison of particles generated by \textbf{(left)} the bootstrap
particle filter, and \textbf{(right)} the bridge particle filter for
the OU example. The point at $(1,0.15)$ in each plot indicates the
observation. Solid vertical lines indicate times at which resampling
is triggered. By introducing weighting and resampling steps at intermediate
times, the bridge particle filter guides particles toward the observation.\label{fig:LinearDriftBridge-filter}}
\end{figure}

\begin{figure}[p]
\begin{centering}
\includegraphics[width=1\textwidth]{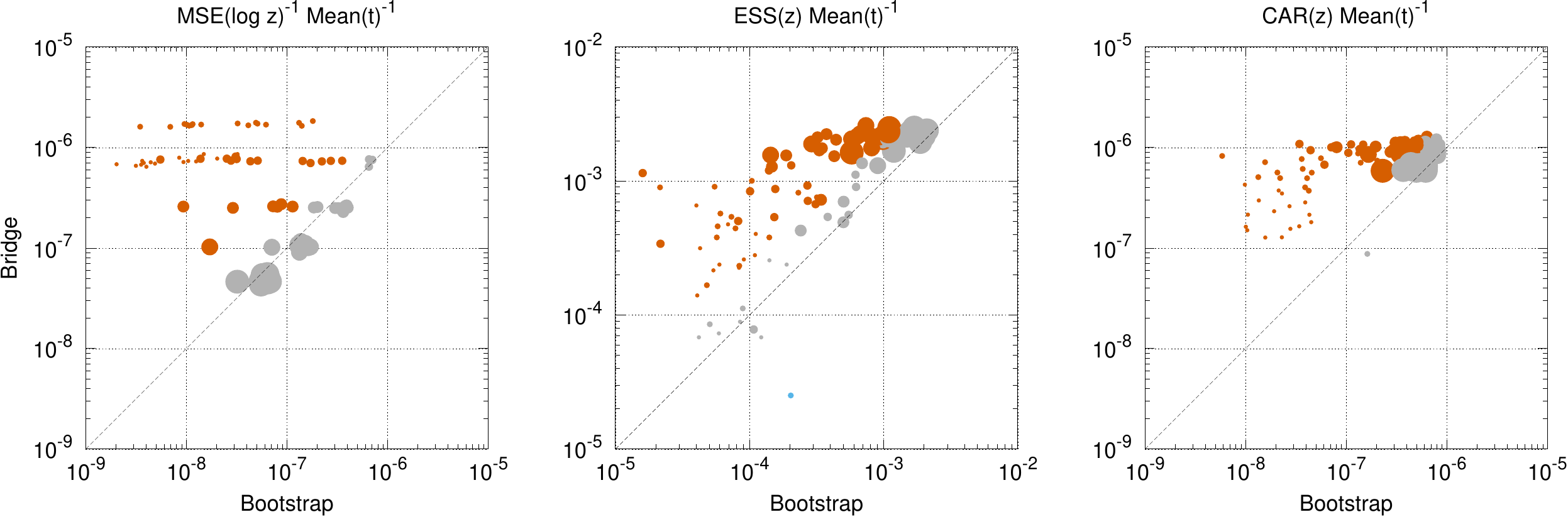}
\par\end{centering}

\protect\caption{Metrics for the OU example, comparing the bootstrap and bridge particle
filters. Ninety-six experiments are conducted, each a unique combination
of one of 16 simulated data sets and a number of particles $N\in\{2^{5},2^{6},\ldots,2^{10}\}$.
The bootstrap particle filter is used to estimate the normalising
constant for each experiment 4096 times, and each of the three metrics
computed from these. The same is done for the bridge particle filter.
Each figure then compares the bootstrap particle filter ($x$-axis)
to the bridge particle filter ($y$-axis) for one of the metrics.
Each point represents one of the experiments, the area of the point
proportional to the number of particles in that experiment. Higher
is better for all metrics, so that points above the diagonal are scenarios
where the bridge particle filter rates superior, while points below
are where the bootstrap particle filter rates superior. Red points
are those where the bridge particle filter outperforms the bootstrap
by a factor of at least two, blue points where the bootstrap particle
filter outperforms the bridge by a factor of at least two. Remaining
points are shown in grey.\label{fig:LinearDriftBridge-metrics}}
\end{figure}

\subsection{Federal Funds Rate (FFR)}

We apply the same process model (\ref{eq:ou-process}) to 25 years
of United States Federal Funds Rate data\footnote{Obtained from http://www.federalreserve.gov/releases/h15/data.htm},
monthly from January 1989 to December 2013, with an interest in parameter
estimation. A similar study is conducted in \citet{Ait-Sahalia1999}.
Algorithm \ref{alg:algorithm1} can be used again. We put prior distributions
on parameters
\begin{eqnarray*}
\theta_{1} & \sim & \mathcal{U}(-1,1)\\
\theta_{2} & \sim & \mathcal{U}(0,1)\\
\theta_{3} & \sim & \mathcal{U}(0,1),
\end{eqnarray*}
where $\mathcal{U}(a,b)$ denotes a uniform distribution on the interval
$[a,b]$.

The first comparison is of the normalising constant estimates of the
bootstrap and bridge particle filters using the three metrics introduced
above. We simulate 16 parameter sets from the prior distribution.
The number of particles is set variously to $N\in\{2^{5},2^{6},\ldots,2^{10}\}$.
Each unique pair of a parameter set and an $N$ constitutes an experiment,
for 96 experiments in total. The bootstrap and bridge particle filters
are applied to each experiment 4096 times, each time producing an
estimate of the normalising constant. Each of the three metrics is
computed from these estimates, for 288 comparisons in total. For computing
the MSE-based metric, the true normalising constant is used, this
being readily computed as the model is linear and Gaussian. Results
are in Figure \ref{fig:FederalFundsRate-metrics}.

The second comparison is to perform parameter estimation using a PMMH
sampler. Two Markov chains are initialised from the same initial state,
obtained from a pilot run that is considered to have converged to
the posterior distribution. The same proposal distribution is used
for both chains. Each chain is configured as in Table \ref{tab:config}.
The posterior distribution for the chain using the bridge particle
filter is given in Figure \ref{fig:FederalFundsRate-posterior}, and
its acceptance rate and effective sample size in Table \ref{tab:config}.
We have been unable to configure the bootstrap particle filter to
work in this example, however. This is explained by the posterior
of $\theta_{3}$ being concentrated on very small values around 0.0017
to 0.0021 (see Figure \ref{fig:FederalFundsRate-posterior}). This
results in a process with very narrow diffusivity, for which, it would
seem, at least some observations become outliers with respect to the
prior. These cannot be tracked by the bootstrap particle filter with
a computationally feasible number of particles. The bridge particle
filter can, even at this low $\theta_{3}$, because of the additional
weighting and resampling steps.

\begin{figure}[p]
\begin{centering}
\includegraphics[width=1\textwidth]{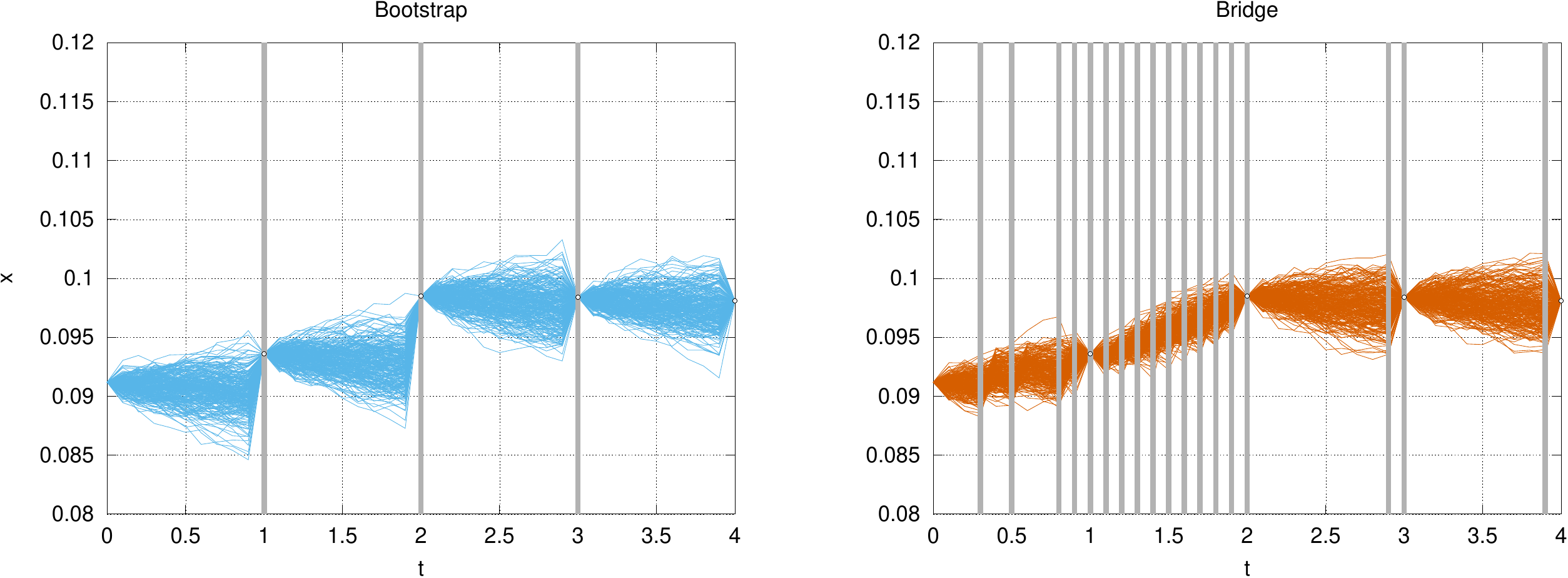}
\par\end{centering}

\protect\caption{Comparison of particles generated by \textbf{(left)} the bootstrap
particle filter, and \textbf{(right)} the bridge particle filter for
the FFR example, over a short interval of time, with a particular
parameter setting. Solid vertical lines indicate times at which resampling
is triggered.\label{fig:FederalFundsRate-filter}}
\end{figure}
\begin{figure}[p]
\begin{centering}
\includegraphics[width=1\textwidth]{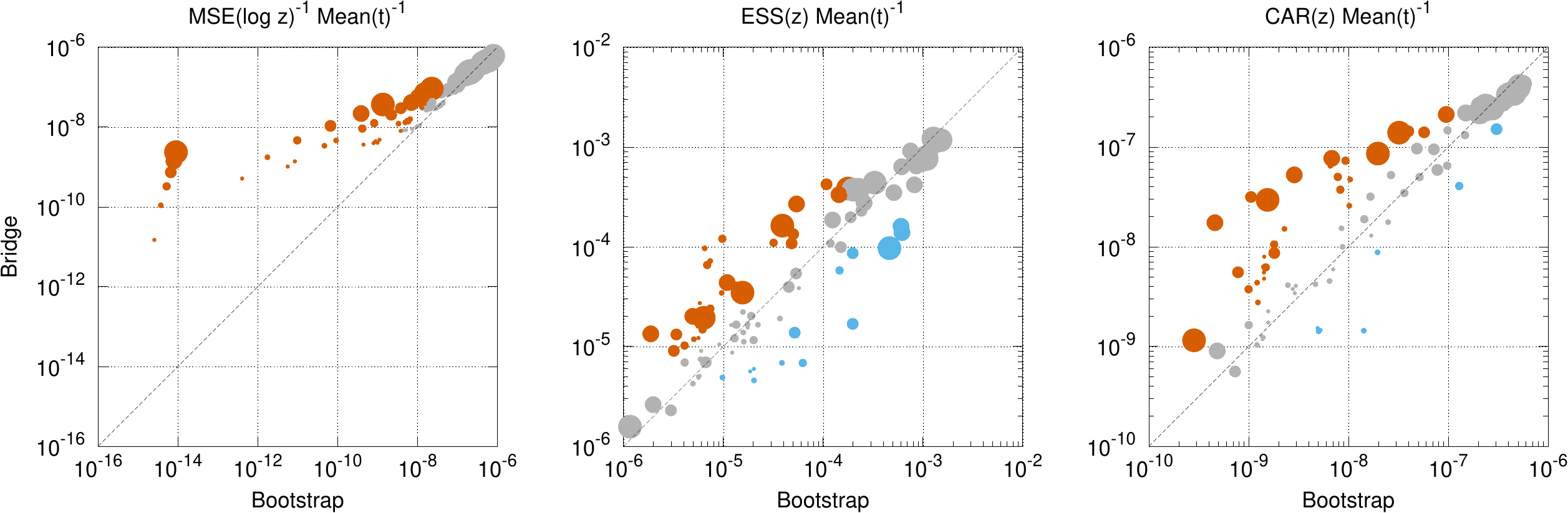}
\par\end{centering}

\protect\caption{Metrics for the FFR example, comparing the bootstrap and bridge particle
filters. See Figure \ref{fig:LinearDriftBridge-metrics} caption for
explanation.\label{fig:FederalFundsRate-metrics}}
\end{figure}

\begin{figure}
\begin{centering}
\includegraphics[width=1\textwidth]{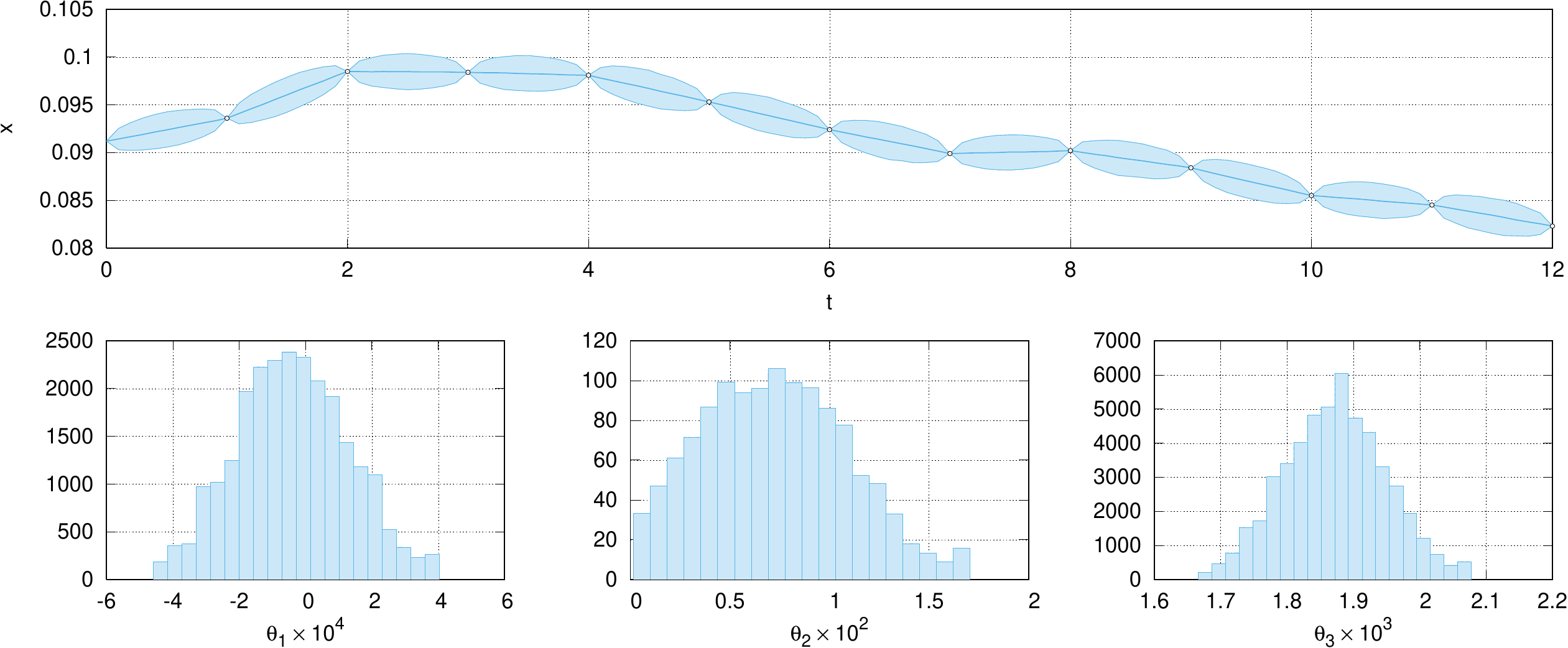}
\par\end{centering}

\protect\caption{Posterior distributions over the state variable (first year only)
and parameters for the FFR example, obtained with PMMH using the bridge
particle filter. For the state variable, the shaded region indicates
the 95\% credibility interval at each time, and the middle line the
median. Points mark observations. For parameters, a normalised histogram
is given.\label{fig:FederalFundsRate-posterior}}
\end{figure}

\subsection{Periodic drift (PD) process}

Consider the diffusion process satisfying the following Itô SDE with
nonlinear drift function, introduced in \citet{Beskos2006} and studied
further in \citet{Lin2010}:
\begin{equation}
dX(t)=\sin\left(X(t)-\theta\right)\,dt+dW(t).\label{eq:pd-process}
\end{equation}
We fix $\theta=\pi$ and $x(0)=0$, and discretise using an Euler--Maruyama
discretisation at a time step of 0.075.

Algorithm \ref{alg:algorithm2} can be used. A Gaussian process does
not capture the dynamics of this process well, as it is multi-modal.
We instead propose a parametric weight function
\begin{equation}
q(x_{n}\,|\,x_{k}):=\frac{1}{z}\left(\cos(x_{n}-\hat{x}_{k})+1+\epsilon\right)\exp\left(-\frac{(x_{n}-\hat{x}_{k})^{2}}{2\sigma^{2}(t_{n}-t_{k})}\right),\label{eq:periodicdrift-weight}
\end{equation}
where $\hat{x}_{k}$ is $x_{k}$ rounded to the nearest multiple of
$2\pi$, $\epsilon$ is a small positive value meant to prevent the
density from being zero at the cosine troughs, and the normalising
constant $z$ is
\[
z=\sqrt{2\pi\sigma^{2}(t_{n}-t_{k})}\left(\exp\left(-\frac{1}{2}\sigma^{2}(t_{n}-t_{k})\right)+1+\epsilon\right).
\]
To obtain the values of the parameters $\epsilon$ and $\sigma^{2}$,
we perform a maximum likelihood estimation using the Nelder--Mead
method and a data set obtained by simulating 10000 paths from (\ref{eq:pd-process})
for 30 units of time. It is important that this function is not too
tight, or we risk high normalising constant estimates that will dramatically
reduce ESS and CAR. For this reason we take the final function to
the power $0.25$. The result is given in Figure \ref{fig:periodicdrift-weight}.

As an initial test we simulate diffusion bridges conditioned on the
data set in \citet{Lin2010}, that is, $x(30)=1.49$, $x(60)=-5.91$,
$x(90)=-1.17$. Both the bootstrap and bridge particle filters are
then applied. The number of particles is set to $N=128$, with the
bridge particle filter applying intermediate weighting and resampling
at time steps of 1. The results are given in Figure \ref{fig:periodicdrift-paths}.
This shows that additional resampling is indeed triggered in the bridge
particle filter on approach to each observation.

We next generate 16 data sets, each constructed by simulating the
model forward for 3000 time units and taking the state at times $0,30,\ldots,3000$.
The number of particles is set variously to $N\in\{2^{5},2^{6},\ldots,2^{10}\}$.
Each unique pair of a data set and $N$ constitutes an experiment,
for 96 experiments in total. The bootstrap and bridge particle filters
are applied to each experiment 4096 times to produce normalising constant
estimates, and each of the three metrics computed from these. For
computing the MSE-based metric, a bootstrap particle filter using
$N=2^{20}$ is used to compute the ``exact'' log-likelihood for
each parameter set. Results are in Figure \ref{fig:periodicdrift-metrics}.

\begin{figure}[p]
\begin{centering}
\includegraphics[width=1\textwidth]{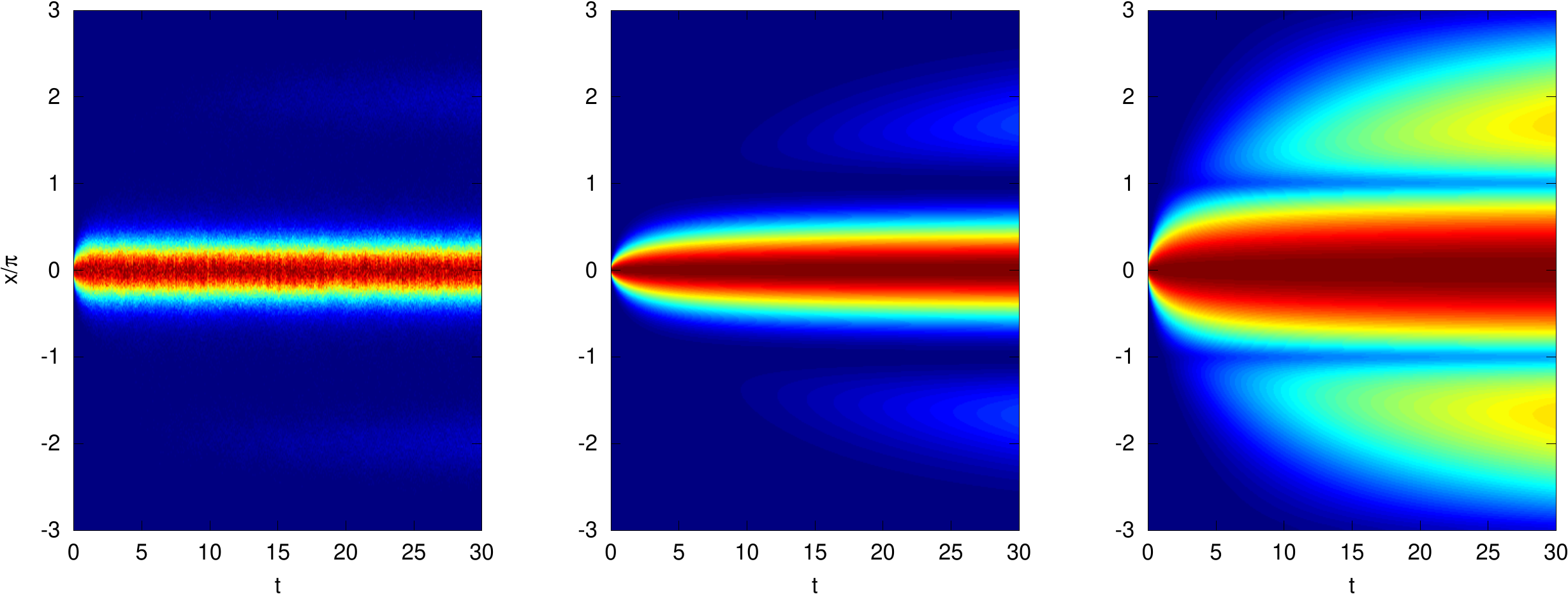}
\par\end{centering}

\protect\caption{Weight function used for the PD example, \textbf{(left)} histogram
of simulated paths used to fit the weight function, \textbf{(centre)}
the weight function (\ref{eq:periodicdrift-weight}) with maximum
likelihood parameter estimates $\epsilon=0.0259$ and $\sigma^{2}=0.3238$,
and \textbf{(right)} the same weight function taken to the power 0.25
(used in the experiments). In all three cases the colour scheme is
scaled between zero and the greatest value in each vertical section.\label{fig:periodicdrift-weight}}
\end{figure}

\begin{figure}[p]
\begin{centering}
\includegraphics[width=1\textwidth]{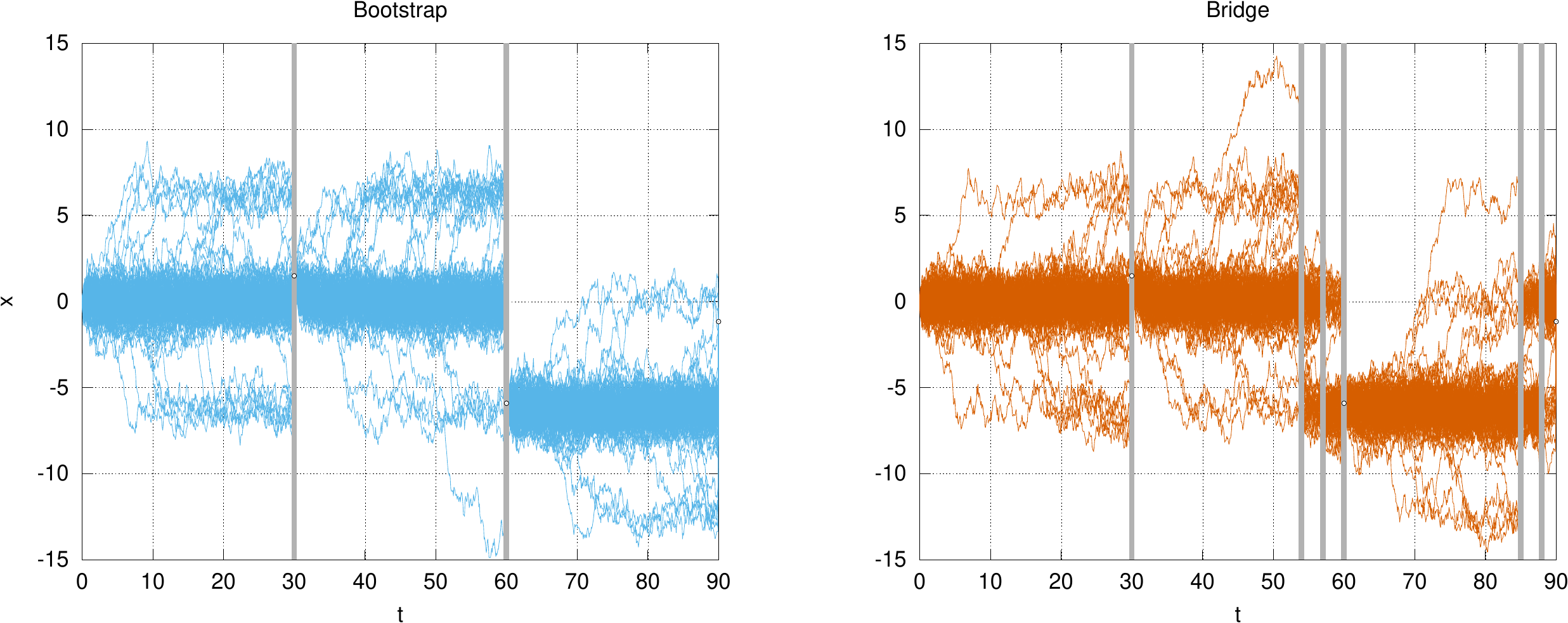}
\par\end{centering}

\protect\caption{Diffusion bridge samples for the PD example with the data set of \citet{Lin2010},
$x(0)=0,\,x(30)=1.49,\,x(60)=-5.91,\,x(90)=-1.17$, using a \textbf{(left)}
bootstrap particle filter, and \textbf{(right)} bridge particle filter.
Solid vertical lines indicate times at which resampling is triggered.\label{fig:periodicdrift-paths}}
\end{figure}

\begin{figure}[p]
\begin{centering}
\includegraphics[width=1\textwidth]{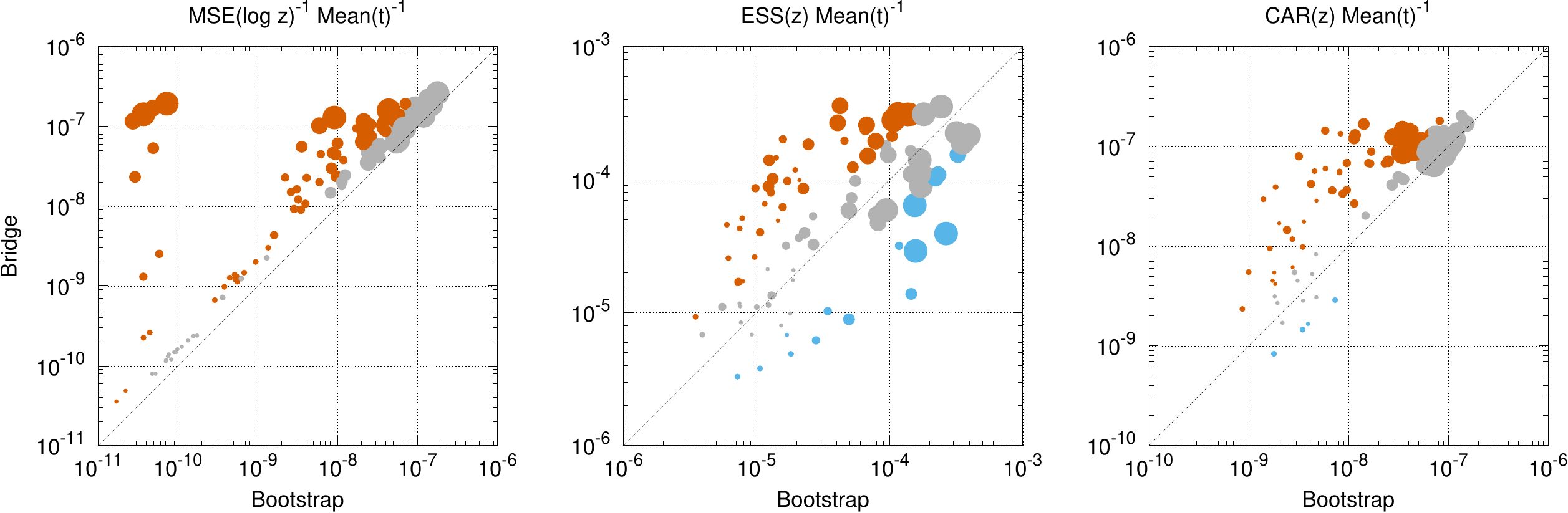}
\par\end{centering}

\protect\caption{Metrics for the PD example, comparing the bootstrap and bridge particle
filters. See Figure \ref{fig:LinearDriftBridge-metrics} caption for
explanation.\label{fig:periodicdrift-metrics}}
\end{figure}

\subsection{Epidemiological SIR model}

Consider an SIR (susceptible/infectious/recovered) model of an epidemic~\citep{Kermack1927},
where $S(t)$ gives the number of susceptible individuals in a population
over time, $I(t)$ the number of infectious individuals, and $R(t)$
the number of recovered individuals, parameterised by an infection
rate $\beta$ and recovery rate $\nu$:
\begin{eqnarray*}
\frac{dS(t)}{dt} & = & -\beta S(t)I(t)\\
\frac{dI(t)}{dt} & = & \beta S(t)I(t)-\nu I(t)\\
\frac{dR(t)}{dt} & = & \nu I(t).
\end{eqnarray*}
We introduce stochasticity into the system by allowing the original
parameters $\beta$ and $\nu$ to vary in time~\citep{Liu2012,Dureau2013},
following the Itô SDEs:
\begin{eqnarray*}
dS(t) & = & -\beta(t)S(t)I(t)\,dt\\
dI(t) & = & \left(\beta(t)S(t)I(t)-\nu(t)I(t)\right)\,dt\\
dR(t) & = & \nu(t)I(t)\,dt\\
d\log\beta(t) & = & \left(\theta_{\beta,1}-\theta_{\beta,2}\log\beta(t)\right)\,dt+\theta_{\beta,3}\,dW_{\beta}(t)\\
d\log\nu(t) & = & \left(\theta_{\nu,1}-\theta_{\nu,2}\log\nu(t)\right)\,dt+\theta_{\nu,3}\,dW_{\nu}(t).
\end{eqnarray*}
Note that $dS(t)+dI(t)+dR(t)=0$, so that total population is conserved,
and that $\beta(t)$ and $\nu(t)$ are always positive. %
For numerical simulation, the SDEs are converted to ODEs with a discrete-time
noise innovation:
\begin{eqnarray*}
\frac{dS(t)}{dt} & = & -\beta(t)S(t)I(t)\\
\frac{dI(t)}{dt} & = & \beta(t)S(t)I(t)-\nu(t)I(t)\\
\frac{dR(t)}{dt} & = & \nu(t)I(t)\\
\frac{d\log\beta(t)}{dt} & = & \theta_{\beta,1}-\theta_{\beta,2}\log\beta(t)+\theta_{\beta,3}\frac{\Delta W_{\beta}}{\Delta t}\\
\frac{d\log\nu(t)}{dt} & = & \theta_{\nu,1}-\theta_{\nu,2}\log\nu(t)+\theta_{\nu,3}\frac{\Delta W_{\nu}}{\Delta t}.
\end{eqnarray*}
where each noise term $\Delta W\sim\mathcal{N}(0,\Delta t)$ is an
increment of the Wiener process over a time step of size $\Delta t$.
These ODEs are then numerically integrated forward using a low-storage
fourth-order Runge--Kutta, with embedded third-order solution for
error control, denoted RK4(3)5{[}2R+{]}C~\citep{Carpenter1994}.

For the purposes of sampling diffusion bridges from the model, we
wish to establish an $\epsilon$-ball around observations within which
samples must fall, where $\epsilon$ is comparable to the discretisation
error in simulating the model forward. It will make matters unnecessarily
difficult to attempt to be any more accurate than this. The RK4(3)5{[}2R+{]}C
algorithm outputs an error estimate for each variable at each time
step, denoted $\epsilon_{S}(t)$, $\epsilon_{I}(t)$ and $\epsilon_{R}(t)$,
computed as the difference between its third and fourth order solutions.
A decison must then be made, according to these errors, whether to
accept or reject the step. The particular implementation~\citep{Murray2012}
in LibBi is based on the description of error control for the DOPRI5
method in \citet{Hairer1993}. It uses an error tolerance parameterised
by $\delta_{\text{abs}}$ (an absolute tolerance) and $\delta_{\text{rel}}$
(a relative tolerance). We set $\delta_{\text{abs}}=10^{-2}$ and
$\delta_{\text{rel}}=10^{-5}$. For $X=\{S,I,R\}$, these are used
to scale the error
\begin{equation}
\frac{\epsilon_{X}(t)}{\delta_{\text{abs}}+\delta_{\text{rel}}\left|X(t)\right|}.\label{eq:sir-error}
\end{equation}
The mean of these scaled errors is required to be less than one for
the step to be accepted. If the step is rejected, the step size is
reduced for a new attempt.

We suggest that it is sensible to use an $\epsilon$-interval on each
observed variable that is commensurate with this discretisation error.
We also note that the error estimate is of the \emph{local} error
of a single step of the numerical integrator, not of the \emph{cumulative}
error, which will be greater. We should therefore consider this estimate
conservative, and may choose to inflate $\epsilon$ accordingly.

Only $I$ is observed in the data set used below. We introduce an
indirect (albeit highly informative) observation $Y_{I}(t)$, with
uniform observation model:
\[
Y_{I}(t)\sim\mathcal{U}\left(I(t)-\epsilon(t),I(t)+\epsilon(t)\right).
\]
We can set $\epsilon(t)=\delta_{\text{abs}}+\delta_{\text{rel}}\left|I(t)\right|$,
using the same error threshold as in (\ref{eq:sir-error}), but this
leads to the peculiar situation where the interval is wider for larger
$I(t)$, so that the model grants higher likelihood to smaller $I(t)$.
This seems undesirable, so we remove the relative component. The overall
population in the data set to be used is 763, so that the maximum
error threshold is $\delta_{\text{abs}}+\delta_{\text{rel}}\times763=0.01763$.
Rounding up, we choose to set a constant $\epsilon=0.02$ for all
$t$. With this in place, the errors in the observation, like those
in numerical integration, are kept accurate to a small fraction of
an individual of the population. We do not claim that this selection
of $\epsilon$ is optimal, only that anything much less is futile
without also tightening the error tolerances on the numerical integrator.

An alternative interpretation of this is as an ABC rejection algorithm
with distance function $\rho(x(t),y(t)):=\left|I(t)-Y_{I}(t)\right|$
and acceptance criterion $\rho(x(t),y(t))\leq\epsilon$, with the
choice of $\epsilon$ guided by the discretisation error of the numerical
integrator.

We use Algorithm \ref{alg:algorithm3} with weight functions derived
from the Gaussian process approach described in §\ref{sec:implementation}.
The variance of the additional observation noise is set to $\varsigma^{2}(t)=\epsilon^{2}$.
The data set records an outbreak of Russian influenza in a boys boarding
school in northern England in 1978~\citep{Anonymous1978}. The data
set is also studied in \citet{Ross2009}.

We put prior distributions over the parameters:
\begin{eqnarray*}
\theta_{\beta,1} & \sim & \mathcal{U}(-100,100)\\
\theta_{\beta,2} & \sim & \Gamma(2,1)\\
\theta_{\beta,3} & \sim & \mathcal{U}(0,100)\\
\theta_{\nu,1} & \sim & \mathcal{U}(-100,100)\\
\theta_{\nu,2} & \sim & \Gamma(2,1)\\
\theta_{\nu,3} & \sim & \mathcal{U}(0,100),
\end{eqnarray*}
and the initial values of state variables:
\begin{eqnarray*}
S(0) & = & 763-Y_{I}(0)\\
I(0) & = & Y_{I}(0)\\
R(0) & = & 0\\
\log\beta(0) & \sim & \mathcal{N}\left(\frac{\theta_{1,\beta}}{\theta_{2,\beta}},\frac{\theta_{3,\beta}^{2}}{2\theta_{2,\beta}}\right)\\
\log\nu(0) & \sim & \mathcal{N}\left(\frac{\theta_{1,\nu}}{\theta_{2,\nu}},\frac{\theta_{3,\nu}^{2}}{2\theta_{2,\nu}}\right).
\end{eqnarray*}
Note that for $\log\beta(0)$ and $\log\nu(0)$, the prior is the
same as the stationary distribution.

In preliminary experiments, we find that the bootstrap particle filter
degenerates frequently for many settings of the parameters, while
the bridge particle filter is much more reliable. Consequently, a
PMMH chain using the bridge particle filter is configured as in Table
\ref{tab:config} and used for a pilot run. This pilot run is continued
until it appears to have converged to the posterior distribution,
and has drawn enough samples to fit an improved random-walk Gaussian
proposal. Initialised from the last state of the pilot run, we find
that the bootstrap particle filter can now work more reliably and
so may be used for a comparison. We run two chains, one using the
bootstrap and the other the bridge particle filter, both initialised
from the last state of the pilot run, and configured as in Table \ref{tab:config}.
Their resulting acceptance rates and effective sample sizes are reported
there also. The chain using the bridge particle filter clearly performs
better, although is slower also, given the additional resampling steps.
The posterior results from this chain are given in Figure \ref{fig:SIR-posterior}.

\begin{figure}[p]
\centering{}\includegraphics[width=1\textwidth]{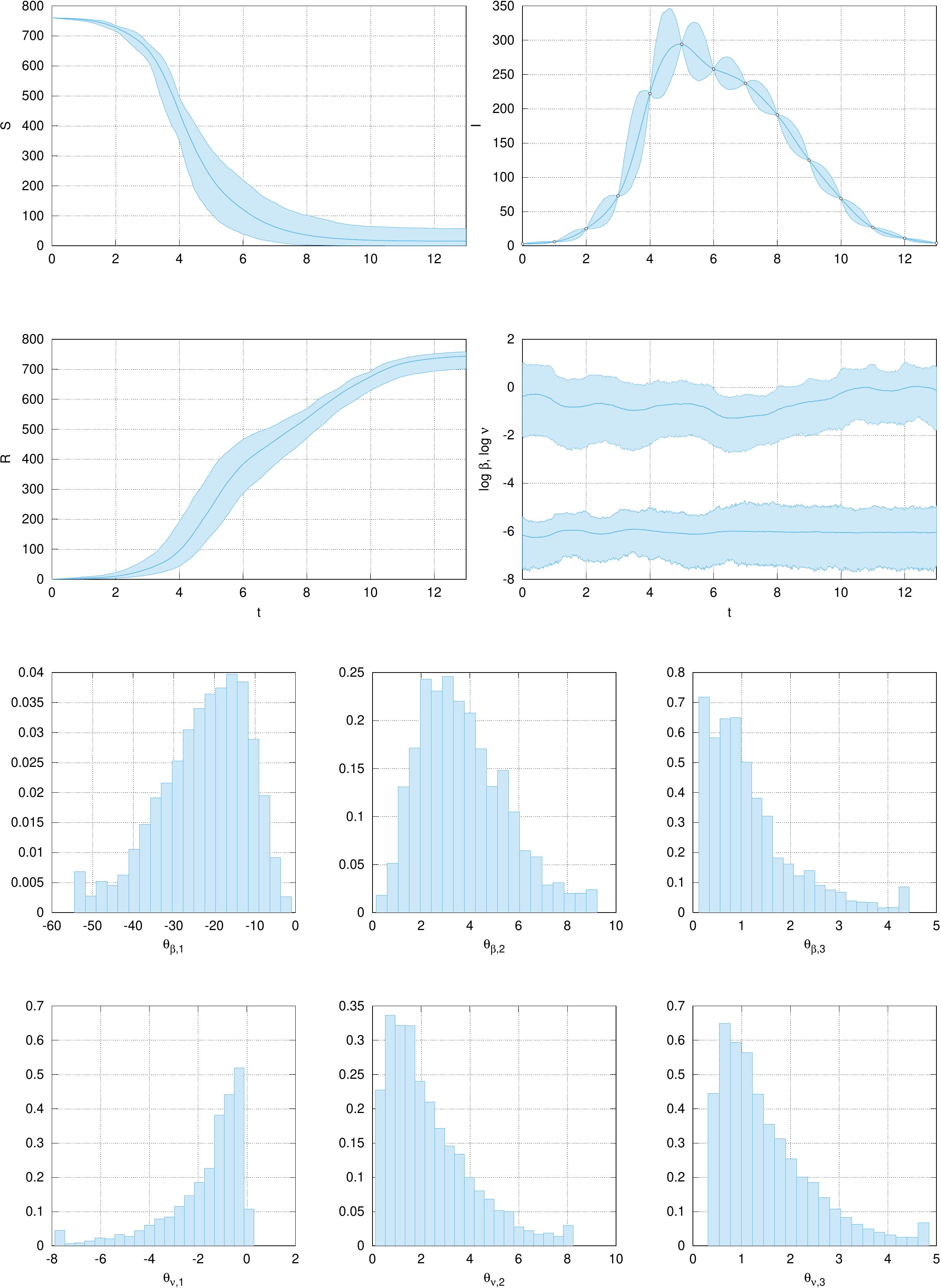}\protect\caption{Posterior distributions over state variables and parameters for the
SIR example, obtained with PMMH using the bridge particle filter.
For state variables, shaded regions indicate 95\% credibility intervals
at each time, and middle lines the medians. Only $I$ is observed,
with observations marked. For parameters, a normalised histogram is
given.\label{fig:SIR-posterior}}
\end{figure}

\subsection{Marine biogeochemical NPZD model}

Marine biogeochemical models are an important means of assessing ecosystem
health, especially in coastal environments. We adopt the \emph{NPZD
}model of \citet{Parslow2013}. The model is described in detail in
that work, and summarised in \citet{Murray2013a}. As even the summary
takes some pages, we give only the high level motivation here.

An NPZD model represents the interaction of nutrients ($N$), phytoplankton
($P$), zooplankton ($Z$) and detritus ($D$) in a body of water.
Each variable represents a compartment of a closed system, its value
representing the quantity of nitrogen contained in that compartment.
The four variables interact via a differential system where
\[
\frac{dN(t)}{dt}+\frac{dP(t)}{dt}+\frac{dZ(t)}{dt}+\frac{dD(t)}{dt}=0,
\]
such that the total quantity of nitrogen (in a closed system) is conserved.
The fluxes between compartments are nonlinear functions of nine stochastic
autoregressive terms, which model various biological, chemical and
physical processes on a discrete-time daily time step. A basic loop
is the absorption of nutrient by phytoplankton during growth, the
grazing of phytoplankton by zooplankton, the death of zooplankton
to produce detritus, and the remineralization of that detritus into
nutrient. The fluxes between variables are not limited to these particular
interactions, however.

The NPZD model is physically positioned somewhere in the open ocean,
within the surface mixed layer. There, it is subjected to exogenous
environmental forcings such as daily temperature and light availability,
and is opened by a bottom boundary condition that permits a flux of
nutrient from below.

The data set used is from the site of Ocean Station P in the north
Pacific. This data set has been studied before in \citet{Matear1995}
and \citet{Parslow2013}. We take four years of data, 1971--1974,
as in \citet{Parslow2013}. Observations include dissolved inorganic
nitrogen ($Y_{N}$), considered an observation of nutrient ($N$),
and chlorophyll-a fluorescence ($Y_{Chla}$), an observation of chlorophyll-a
($Chla$), itself a state variable that is a function of phytoplankton
($P$) and available light.

The observation model is:
\begin{eqnarray*}
\log Y_{N}(t) & \sim & \mathcal{N}(\log N(t),0.2)\\
\log Y_{Chla}(t) & \sim & \mathcal{N}(\log Chla(t),0.5).
\end{eqnarray*}
While this observation model may appear only weakly informative, it
can become highly informative given the sparsity (in time) at which
observations are available. While the state variables tend to vary
on daily or weekly time scales, the largest gap in observations of
$Y_{N}$ is 136 days, and that of $Y_{Chla}$ 101 days.

We use Algorithm \ref{alg:algorithm3}, with weight functions derived
from the Gaussian process approach described in §\ref{sec:implementation}.
Gaussian processes are fit to the logarithm of the observed time series
of $Y_{N}$ and $Y_{Chla}$.

The first comparison is of the normalising constant estimates of the
bootstrap and bridge particle filters, using the three metrics introduced
above. We simulate 16 parameter sets from the prior distribution.
The number of particles is set variously to $N\in\{2^{5},2^{6},\ldots,2^{10}\}$.
Each unique pair of a parameter set and an $N$ constitutes an experiment,
for 96 experiments in total. The bootstrap and bridge particle filters
are applied to each experiment 1024 times, each time producing an
estimate of the normalising constant. Each of the three metrics is
computed from these estimates, for 288 comparisons in total. For computing
the MSE-based metric, a bootstrap particle filter using $N=2^{20}$
is used to compute the ``exact'' log-likelihood for each parameter
set. Results are in Figure \ref{fig:NPZD-metrics}.

The second comparison is to perform parameter estimation using a PMMH
sampler. Two Markov chains are initialised from the same initial state,
obtained from a pilot run that is assessed to have converged to the
posterior distribution. The same proposal distribution is used for
both chains. They are otherwise configured as in Table \ref{tab:config},
where their acceptance rates and effective sample sizes are also reported.
The posterior distribution for the chain using the bridge particle
filter is given in Figure \ref{fig:NPZD-posterior}.

\begin{figure}[p]
\begin{centering}
\includegraphics[width=1\textwidth]{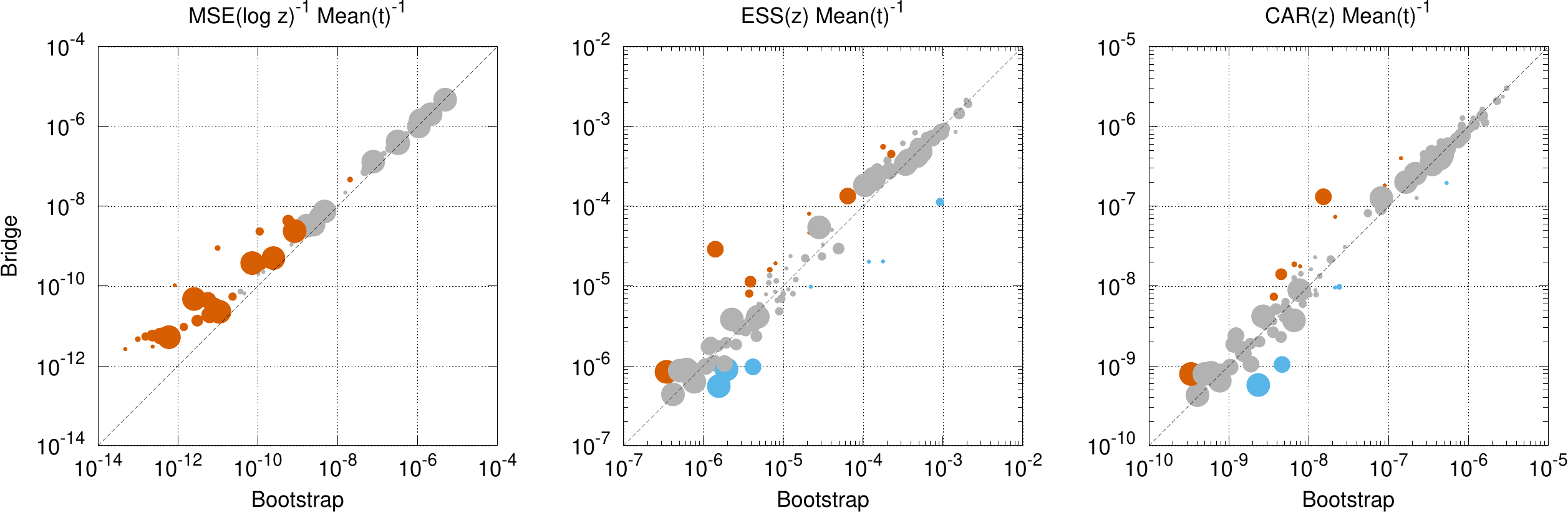}
\par\end{centering}

\protect\caption{Metrics for the NPZD example, comparing the bootstrap and bridge particle
filters. See Figure \ref{fig:LinearDriftBridge-metrics} caption for
explanation.\label{fig:NPZD-metrics}}
\end{figure}

\begin{figure}[p]
\begin{centering}
\includegraphics[width=1\textwidth]{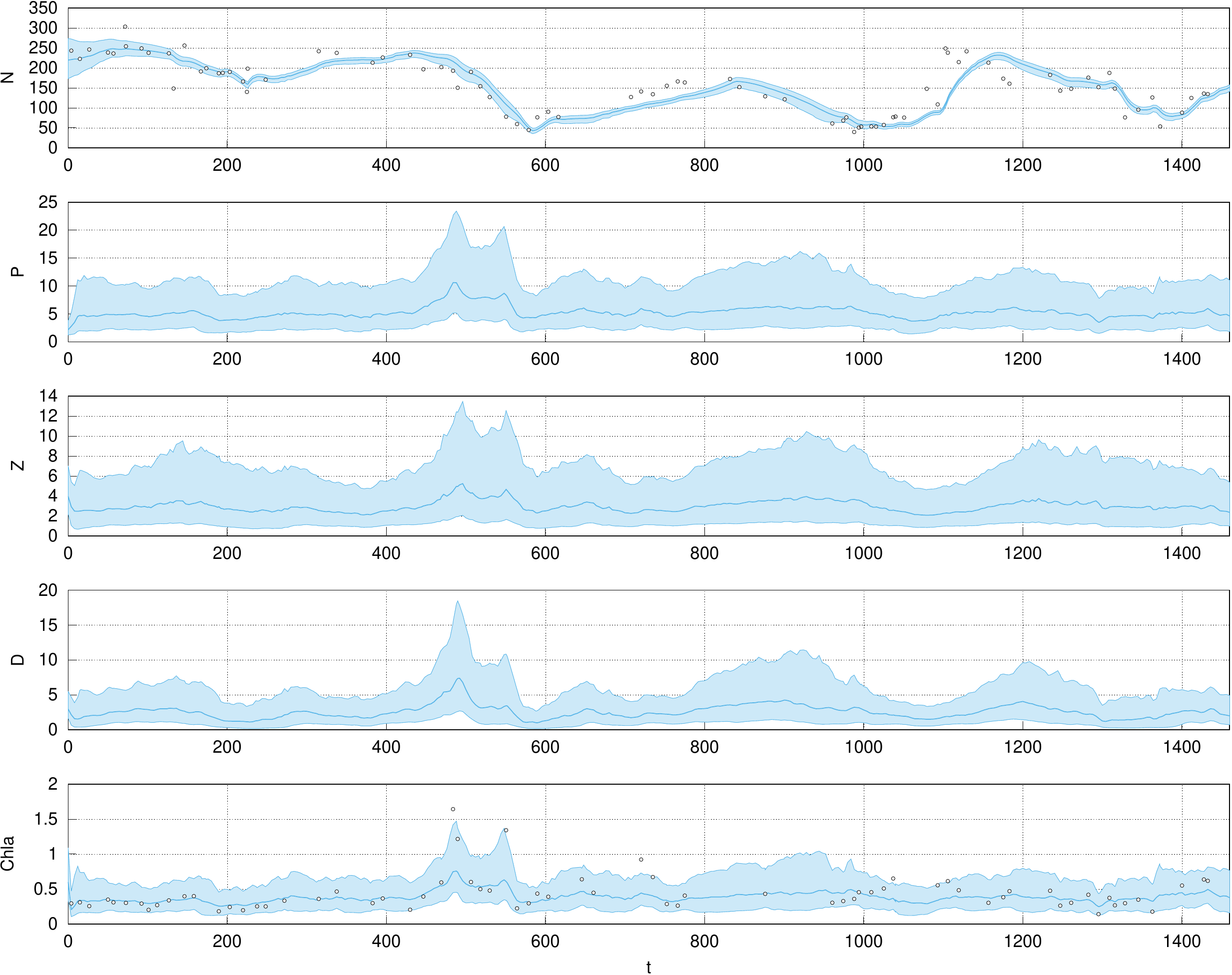}
\par\end{centering}

\protect\caption{Posterior distributions over state variables for the NPZD example,
obtained with PMMH using the bridge particle filter. Shaded regions
indicate 95\% credibility intervals at each time, and middle lines
the medians. Both $N$ and $Chla$ are observed, with observations
marked; others are unobserved.\label{fig:NPZD-posterior}}
\end{figure}

\section{Discussion\label{sec:discussion}}

Of the four examples where the bootstrap and bridge particle filters
are compared on metrics (OU, FFR, PD, NPZD), the bridge particle filter
is consistently superior on the MSE metric, and at least as good on
the ESS and CAR metrics. ESS is much better on the OU example, but
similar on the others. CAR is much better on the OU example, modestly
better on the FFR and PD examples, and similar on the NPZD example.
This is evident in Figures \ref{fig:LinearDriftBridge-metrics}, \ref{fig:FederalFundsRate-metrics},
\ref{fig:periodicdrift-metrics} \& \ref{fig:NPZD-metrics}. Recall
that these metrics are already adjusted for execution time, which
is typically longer for the bridge particle filter due to additional
weighting and resampling steps.

Of the three examples where a comparison of PMMH performance using
a real data set was attempted (FFR, SIR and NPZD), we were unable
to find a working configuration for the bootstrap particle filter
for the FFR example, and had difficulties configuring it for the SIR
example due to frequent degeneracy. The bridge particle filter worked
reliably in both cases, however. On the SIR and NPZD examples, the
bridge particle filter outperforms the bootstrap on both ESS and acceptance
rate (see Table \ref{tab:config}). This suggests that the bridge
particle filter can work well in cases where the bootstrap particle
filter does not.

We can report that the process of pilot runs---for tuning the number
of particles and proposal distribution---was less unpleasant than
usual when using the bridge particle filter. For the SIR example,
the bootstrap particle filter could not be made to work for such pilot
runs, so that the bridge particle filter was required for this purpose.
While anecdotal, this experience does affirm that the additional weighting
and resampling steps are useful, and may compensate for a poor setting
of parameters, or poorly fitting model.

There are a number of areas where care is needed in configuring the
bridge particle filter:
\begin{enumerate}
\item The additional resampling steps can decrease performance, by introducing
additional variance in the normalising constant estimate. The use
of an adaptive resampling trigger (such as the ESS used here) mitigates
this. In the empirical results of this work, the additional resampling
steps appear to have a net benefit, or in the worst cases, do no harm.
\item The weight functions used may be too tight, so that particles are
selected too aggressively at intermediate resamplings. For the PD,
SIR and NPZD examples we have taken the weight function to the power
1/4 as a precaution. This seems sufficient for the examples here,
but a more rigorous approach might be to use heavier-tailed weight
functions (e.g. a Student $t$).
\end{enumerate}
Finally, we have used a schedule of equispaced times for the additional
weighting and resampling steps. This may be wasteful of compute resources.
Alternative schedules may be superior, such as a geometric series
of decreasing interval length, so that resampling is more frequent
on approach to the observation. The schedule may even be adapted.
We have found these ideas unecessary to pursue for the examples here,
however, and so leave them to future work.

\section{Conclusion\label{sec:conclusion}}

This paper has presented three related SMC methods for handling state-space
models with highly-informative observations, including the special
case of sampling bridges between fixed initial and final values. We
have referred to them collectively as bridge particle filters. These
bridge particle filters appear to improve substantially on the bootstrap
particle filter in terms of the MSE of normalising constant estimates,
and more modestly on other metrics. For two applications, we have
been able to apply a PMMH sampler using bridge particle filters when
we either could not do so, or had difficulty doing so, with a bootstrap
particle filter, and we report anecdotally that these methods are
quite straightforward to configure.

\section*{Supplementary material}

All examples are available for download from the LibBi website (\url{www.libbi.org}).

\bibliographystyle{abbrvnat}
\bibliography{../../work/bib/bgc}

\end{document}